\documentclass{amsart}

\usepackage{amssymb}
\usepackage[dvipsnames]{xcolor}
 \usepackage{graphics}
\usepackage{tikz}
\usetikzlibrary{arrows.meta}

\usepackage{epsfig}  		
\usepackage{epic,eepic}       

\usepackage{cite}

\newtheorem{thm}{Theorem}[section]

\newtheorem{lem}[thm]{Lemma}

\theoremstyle{definition}
\newtheorem{definition}[thm]{Definition}
\newtheorem{example}[thm]{Example}

\theoremstyle{remark}

\newtheorem{remark}[thm]{Remark}

\numberwithin{equation}{section}

\begin{document}


\title[The Interconnectivity Vector]{The Interconnectivity Vector: A Finite-Dimensional Vector Representation of Persistent Homology}


\author{Megan Johnson}
\address{Department of Mathematics, University at Buffalo, SUNY,
Buffalo, NY 14260}
\email{meganjoh@buffalo.edu}


\author{Jae-Hun Jung}
\address{Department of Mathematics, POSTECH,
Pohang, Korea}
\email{jung153@postech.ac.kr}



\begin{abstract}
Persistent Homology (PH) is a useful tool to study the underlying structure of a data set. Persistence Diagrams (PDs), which are 2D multisets of points, are a concise summary of the information found by studying the PH of a data set. However, PDs are difficult to incorporate into a typical machine learning workflow. To that end, two main methods for representing PDs have been developed: kernel methods and vectorization methods. In this paper we propose a new finite-dimensional vector, called the {\it interconnectivity vector}, representation of a PD adapted from bag-of-words (BoW). This new representation is constructed to demonstrate the connections between the homological features of a data set. This initial definition of the interconnectivity vector proves to be unstable, but we introduce a stabilized version of the vector and prove its stability with respect to small perturbations in the inputs.  We evaluate both versions of the presented vectorization on several data sets and show their high discriminative power.

\end{abstract}


\maketitle

\section{Introduction}

Topological Data Analysis (TDA) is a rich field containing a powerful framework for the analysis of high-dimensional data structure \cite{carlsson:2009}. One of the key tools in TDA is persistent homology (PH), introduced by Edelsbrunner et al. in 2002 \cite{Edelsbrunner2002}. PH provides a detailed snapshot of the geometry of the underlying data's structure. This relatively new field has recently seen an explosion in its possible applications with new applications being explored constantly. 
Previous applications of TDA include 3D shape segmentation \cite{carriere:3Dshapes}, astronomy \cite{Xu2019, heydenreich2020}, biology and medicines \cite{McGuirl2020,Nicponski2020}, contagious disease spread \cite{Taylor:2015, Lo2018}, data mining \cite{Almgren:data}, finance \cite{gidea:2018, Gidea2017}, image tampering detection \cite{asaad:2017}, material science \cite{Lee2018, Sorensen2020}, neuroscience \cite{bendich2016, Sizemore:2019}, social networks \cite{Almgren2017}, and even sports analytics \cite{Goldfarb2014}.

The results of TDA, whether in the form of a persistence diagram (PD) or a barcode, do not lend themselves to inclusion in a typical machine learning workflow as they are multisets and graphs, respectively. To reconcile TDA with machine learning algorithms, two main avenues have been explored: kernel methods, in which distances between PDs are defined and exploited, and vectorization methods, in which PDs (or barcodes) are sampled or transformed into vectors of finite dimension. 

Kernel methods, such as the Wasserstein Kernel or the Sliced Wasserstein Kernel \cite{carrire2017sliced,kolouri2015sliced}, while theoretically robust, use computationally expensive metrics to compute the distance between PDs. Additionally, kernel methods for PDs treat each point on a PD as if it were alone and do not take into account potential relationships between points, especially the concurrence of topological features represented by these points. Historically, topological spaces were distinguished by their $n$th Betti numbers, the rank of their $n$th homology groups (see Section 2), which roughly count the number of $n$-dimensional holes in the space. Kernel methods using the Wasserstein or Sliced Wasserstein distances usually measure mainly more the physical differences between PDs rather than what those differences mean for the topology of the underlying space.

Vectorization methods, such as the Topological Vector \cite{carriere:3Dshapes} and Persistence Images \cite{adams2015persistence}, are, in general, more computationally efficient than kernel methods but sometimes lack robustness. The construction of the topology vector is straightforward and easy to implement but it can have issues distinguishing between topologies when used in classification problems. If we wish to make distinctions between data sets based on topology we need a vector or method which will contain information (e.g. Betti numbers) about the coexistence of features and the relative size (i.e. length of existence) of features. 

In this paper we introduce a new finite-dimensional vectorization, called the {\it interconnectivity vector}, of the PD which contains information about not only the relative size of each homological feature, but also its relative position in regards to other homological features. In effect, the interconnectivity vector is a measure of the concentration of persistence points that exist during the lifetime of each off-diagonal point on the PD. The original topology vector shows the hierarchical pattern of persistences (or existences) of homological features but not how those features are related. The new interconnectivity vector basically follows the similar idea of the topology vector but it makes the topology vector deliver richer information about the given data structure by considering the interconnectivity between features. In this paper, we will provide detailed analysis of the proposed vectorization method including stability. The initial definition of the interconnectivity vector proves to be unstable meaning that a small change in a PD leads to a large change in the $L_{\infty}$-norm of the vector.  To remedy this we propose a new stabilized version of the interconnectivity vector and prove its stability. In  \cite{zielinski2018persistence}, the bag of words concept, commonly used in natural language processing \cite{Sivic},  was used for the construction of persistence bag of words as the PD is a multiset, and the Gaussian mixture model was applied to stabilize the persistence bag of words. We adopt the similar Gaussian approach to stabilize the interconnectivity vector. We show numerically that the original interconnectivity vector and stabilized interconnectivity vector yield different behaviors in their realization. Numerical examples show the validation of the proposed vectorization.

This paper will be organized as follows: Section 2 will cover the background of PH and its representations. Section 3 will cover previous kernel and vectorization methods for PDs. Section 4 will introduce the definition of the interconnectivity vector as well as provide examples and discuss stability. Section 5 introduces a stabilized version of the interconnectivity along with an example. Section 6 demonstrates the effectiveness of the new vectorization on classical data sets. Section 7 contains our concluding remarks.

\section{Background}
In this section we will briefly discuss homology, persistent homology, and  representations of persistent homology as persistence diagrams and barcodes.  To begin, we will start with homology groups of arbitrary topological spaces and then introduce homology groups of simplical complexes before finishing with persistent homology and its representations. For more information see \cite{hatcher:2000, carlsson:2009, Otter2015roadmap}.

\subsection{Singular Homology}

In algebraic topology, homology is a topological invariant used to, roughly speaking, describe the holes of various dimensions in a given space. Recall that a simplex is a generalized notion of a triangle or tetrahedron. For example, a 0-simplex is a vertex or point, a 1-simplex is a line or edge, a 2-simplex is a triangle, a 3-simplex is a tetrahedron, and so on and so forth.

\begin{definition} An $n$-simplex is the smallest convex set in a Euclidean space $\mathbb{R}^m$ containing the $n+1$ points $v_0,v_1, \dots, v_n$ (called the vertices) that do not lie in a hyperplane of dimension less than $n$. The simplex is denoted by $[v_0,v_1,\dots,v_n]$.
\end{definition}

\begin{definition} The standard $n$-simplex is
$$\Delta^n = \{ (t_0,t_1, \dots,t_n) \in \mathbb{R}^{n+1} | \, \Sigma_i t_i = 1 \text{ and } t_i \geq 0 \text{ for all } i\}$$
whose vertices are unit vectors on the coordinate axes.
\end{definition}

We will now define singular homology group for a topological space $X$ but first we need the following definition
\begin{definition} A singular $n$-simplex in a topological space $X$ is a continuous map $\sigma$ from the standard $n$-simplex, $\Delta^n$ to $X$, i.e. $\sigma: \Delta^n \to X$. Note that $\sigma$ need not be injective. \end{definition}

Let $C_n(X)$ be the free abelian group whose basis is the set of singular $n$-simplices in $X$. Elements of $C_n(X)$ are finite formal sums $\sum_{i} n_i \sigma_i$ for $n_i \in \mathbb{Z}$ and $\sigma_i : \Delta^n \to X$. Define the boundary map $\delta_n : C_n(X) \to C_{n-1}(X)$ by 
$$\delta_n (\sigma) = \sum_i (-1)^i \sigma | \, [v_0, \dots,\hat{v}_i, \dots v_n]$$
where the notation $\hat{v}_i$ means that the vertex $v_i$ is removed and $\sigma | \, [v_0, \dots,\hat{v}_i, \dots v_n]$ is the restriction of $\sigma$ to the remaining vertices. Note that this restriction of $\sigma$ to the remaining $n$ vertices is itself a singular $(n-1)$-simplex as it is a map from $\Delta^{n-1} \to X$. 
\begin{lem}\label{boundary} The composition $ \delta_{n-1} \circ \delta_{n}: C_{n} \to C_{n-1} \to C_{n-2}$ is trivial.
\end{lem}
(For the proof of Lemma \ref{boundary} see \cite{hatcher:2000}.)

This vanishing consecutive boundary relation gives us the inclusion $Im~ \delta_{n+1} \subset$ $Ker~  \delta_n$, where $Im$ and $Ker$ denote the image and kernel, respectively. Finally, we have

\begin{definition}The $n$-th homology group of the topological space $X$ is the quotient group
$$H_n(X) = \text{ Ker } \delta_n / \text{ Im } \delta_{n+1}$$

\end{definition}

Roughly speaking, the rank of the $n$-dimensional homology group indicates how many {\it holes} of dimension $n$ exist in $X$. Specifically, the $0$-dimensional homology group counts the number of connected components in $X$, the $1$-dimensional homology group counts the number of $S^1$ type holes, the $2$-dimensional homology group counts the number of vacuous or $S^2$ type holes, etc. Take for example the unit sphere. For the sphere we have $H_0 = \mathbb{Z}^1, H_1 = \mathbb{Z}^0, H_2 = \mathbb{Z}^1$. On the other hand, for a torus we have $H_0 = \mathbb{Z}^1, H_1 = \mathbb{Z}^2, H_2 = \mathbb{Z}^1$. The existence, quantity, and type of holes are all useful information that we will use in this research. 


\subsection{Simplicial Homology}
It can be computationally difficult to compute the homology of arbitrary topological spaces so instead we often approximate our space by simplical complexes (rather than $\Delta^n$ and singular simplices) for which homology can be computed algorithmically. 

\begin{definition}
A simplicial complex $K$ is a collection of simplices such that 
\begin{itemize}
\item if $\sigma$ is a simplex in $K$, then $K$ contains all lower-dimensional simplices of $\sigma$, and
\item the non-empty intersection of any two simplices in $K$ is a simplex in $K$. 
\end{itemize}
The dimension of the simplicial complex $K$ is the maximum of the dimension of its simplices. A map between simplicial complexes, $f: K \to L$, is a map that sends the vertex set of $K$ to the vertex set of $L$ such that $f(\sigma) \in L$ for all $\sigma \in K$. 
\end{definition}

Let $K$ be a simplical complex and let $C_n(K)$ be the free abelian group whose basis is the set of $n$-simplices in $K$. Elements of $C_n(K)$ are finite formal sums $\sum_{i} n_i \sigma_i$ for $n_i \in \mathbb{Z}$ and where $\sigma_i$ is a $n$-simplex in $K$.

Let $\sigma$, a basis element of $C_n(K)$, be an oriented $n$-simplex with vertices $v_0, \dots, v_n$.  Let $\delta_n:$ $C_n \rightarrow C_{n-1}$ be the boundary map 
$$\delta_n(\sigma) = \sum_{i=0}^n (-1)^i[v_0,\dots, \hat{v}_i, \dots, v_n]$$
where the simplex $[v_0,\dots, \hat{v}_i, \dots, v_n]$ is the $i$th face of $\sigma$ obtained by deleting $v_i$. A slight change to Lemma \ref{boundary} will again show that $\delta_{n-1}\circ \delta_n (\sigma) = 0$. Again this vanishing consecutive boundary relation gives us the inclusion $Im~ \delta_{n+1} \subset $ $Ker~ \delta_n$ and so we have the following definition. 

\begin{definition} The $n$-th homology group of the simplicial complex $K$ is the quotient group
$$H_n(K) = \text{ Ker } \delta_n / \text{ Im } \delta_{n+1}$$

\end{definition}

To quantify the number of $n$-dimensional holes in the topological space $X$, or its simplical complex approximation $K$, we use the $n$-dimensional Betti number, $\beta_n$. 
\begin{definition}The $n$-dimensional Betti number, $\beta_n$, is the rank of the $n$th homology group, i.e. the number of linearly independent generators of $H_n$.
\end{definition}

More explicitly, $\beta_0$ counts the number of connected components, $\beta_1$ counts the number of loops, $\beta_2$ counts the number of vacuous volumes, and so on. Most importantly, Betti numbers are topological invariants, meaning that spaces which are topologically equivalent have the same Betti numbers. 


\subsection{Persistent Homology}
In order to study the homology of a data set we must first construct a simplicial complex using its points as vertices. There are several ways to build this simplicial complex, but in this paper we will restrict ourselves to the {\it Vietoris-Rips simplicial complex}.

\begin{definition}
Given a finite data set $S$ viewed as a subset of the metric space $(X,d)$ and a filtration value $\tau \geq 0$, the Vietoris-Rips complex $VR_{\tau}(S)$ built on the vertex set $S$ has a $k$-simplex for every collection of $k + 1$ vertices that are within a distance $\tau$ of each other. More explicitly, 

$$VR_{\tau}(S) = \{\sigma \subseteq S \, | \, d(x,y) \leq 2\tau \text{ for all } x,y \in \sigma \}$$
\end{definition}

Note that we are only able to construct the Rips complex at a specific scale $\tau$. The choice of $\tau$ is not obvious for any given data set and so persistent homology \cite{Edelsbrunner2002} solves this issue by computing the homology of many nested complexes at various filtration values. Specifically, we create a nested chain of Vietoris-Rips complexes for different values of $\tau$. By virtue of its definition as a simplicial complex, the Vietoris-Rips complex has the nice property that for any increasing sequence of non-negative values $\tau_0 \leq \tau_1 \leq \dots \leq \tau_m$ where $\tau_m$ is the maximum filtration we have
$$VR_{\tau_0} \subseteq VR_{\tau_1} \subseteq \dots \subseteq VR_{\tau_m}= VR_{\tau}$$ 
We call $VR_{\tau}$ a {\it filtered simplicial complex}. 

\begin{figure}[!htb]
\minipage{0.245\textwidth}
  \includegraphics[width=\linewidth]{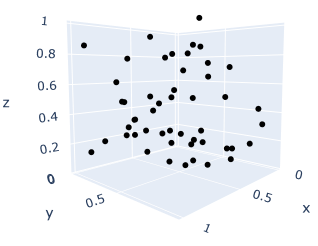}
  
\endminipage\hfill
\minipage{0.245\textwidth}
  \includegraphics[width=\linewidth]{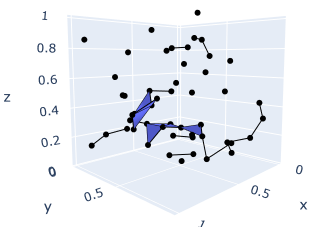}

\endminipage\hfill
\minipage{0.245\textwidth}
  \includegraphics[width=\linewidth]{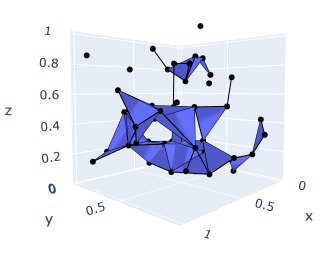}

\endminipage\hfill
\minipage{0.245\textwidth}
  \includegraphics[width=\linewidth]{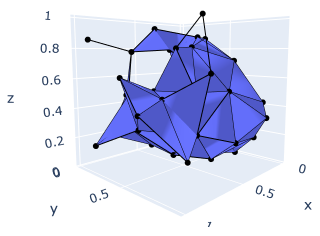}
\endminipage\hfill

\caption{Vietoris-Rips complex for filtration values $\tau =0, 0.2,0.3,0.4$, respectively. Images generated with Plotly for Python.} \label{fig:complex}
\end{figure}

Consider a uniform random distribution of 50 points in the unit cube. Figure \ref{fig:complex} shows four subcomplexes of the Vietoris-Rips filtered simplicial complex on this uniform random point cloud. Note that when a 2-simplex is generated its convex hull is shaded in. Using filtered simplicial complexes we can study both the large and small scale structure of a data set by calculating homology groups.

For PH we compute the $n$-th homology group for each Vietoris-Rips complex in the chain. Note that for all $i,j \in \{0,1, \dots, m\}$ with $i \leq j$ we get a natural inclusion map $VR_{\tau_i} \hookrightarrow VR_{\tau_j}$ which induces a homomorphism $f_{i,j}: H_n(VR_{\tau_i}) \to H_n(VR_{\tau_j})$ between $n$-th homology groups \cite{Otter2015roadmap}. Evaluating how the homology groups, more specifically Betti numbers, change from subcomplex to subcomplex lets us examine the structure of the data set at various scales. The PH of a data set gives us not only the homology groups and the Betti numbers, but also the filtration values for which various topological features appear and disappear. With all this information we need a way to efficiently visualize and analyze the results

\subsection{Representations of Persistent Homology}
There are several ways to represent the information provided by PH but the main representations are persistence barcodes and persistence diagrams. A  {\it persistence barcode} is a graph whose horizontal axis lists filtration values and which contains rays or line segments that correspond to the homology group generators beginning at some filtration {\it birth} and ending at some filtration {\it death} value if its lifetime is finite.
\begin{figure}[hbt!]
  \centering
    \includegraphics[width=0.49\textwidth]{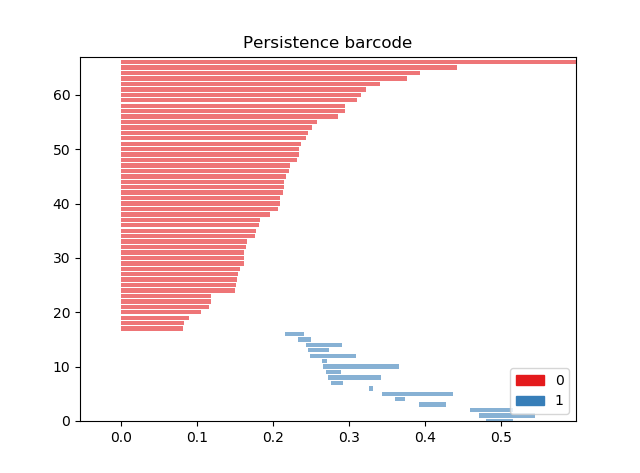}
    \includegraphics[width=0.49\textwidth]{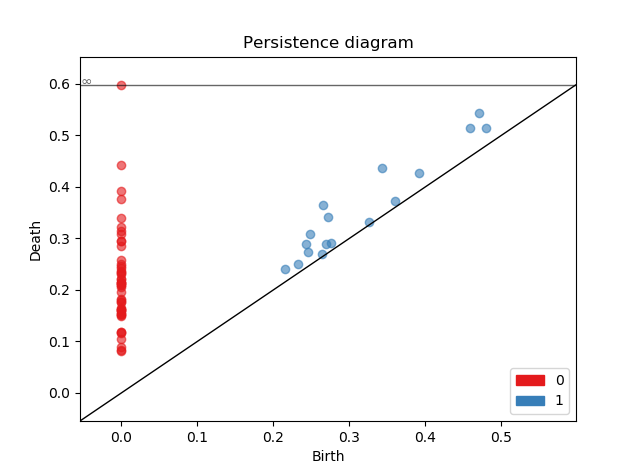}
\caption{The persistence barcode (left) and the persistence diagram (right) for a random distribution of points in the unit cube. Note that  the red bars or points correspond to the 0-dimensional PH and the blue bars or points correspond to the 1-dimensional PH. Additionally, note that the horizontal bar at the top of the PD represents homological features which persist longer than the maximum filtration used in the computation. }\label{fig:barcode}
\end{figure}
In Figure \ref{fig:barcode}, generated using the GUDHI Library (see \cite{gudhi:urm}), we have the barcode for the PH on a randomly distributed set of 50 points in the unit cube with maximum filtration $\tau_m = 0.8$. 

Corresponding to each barcode there is a {\it persistence diagram} (PD), a multiset of points where each point on the diagram represent the birth-death pairs of each generator of the homology group (see Figure \ref{fig:barcode}). The diagonal line $y = x$ is included in the diagram for calculation purposes. Persistence points close to the diagonal have small persistence and generally thought of as noise. Traditionally, birth values are along the horizontal axis and death values are along the vertical axis. 

\section{Machine Learning Representations of Persistent Homology}\label{section:mlrep}

Machine learning algorithms and persistence barcodes or diagrams are incongruous in general unless they are transformed properly. In order to use the information given by PH we need an efficient and topologically robust method for transforming a PD into something compatible with a traditional machine learning workflow. 

Kernel methods for PDs take advantage of theoretically rich, yet computationally costly, metrics such as the Wasserstein or Sliced Wasserstein metrics to compute the ``effort"  of transforming one PD into another \cite{wasserstein:1969,carrire2017sliced}. These metrics measure the physical differences in persistence diagrams, not the underlying information provided by the diagrams. Additionally, these metrics are highly dependent on the scale of the data. There are instances where topologically similar data sets lead to PDs between which the Wasserstein and Sliced Wasserstein distances are relatively large. For example, consider the two point clouds formed by the space curves $r_1(t) = \langle \cos(t),\sin(t),\cos(t^2)\rangle$ and $r_2(t) = \langle \cos(t),\sin(t),\cos(t^2) + \sin(t^2) \rangle$
\begin{figure}[hbt!]
  \centering
\minipage{0.329\textwidth}
    \includegraphics[width=\linewidth]{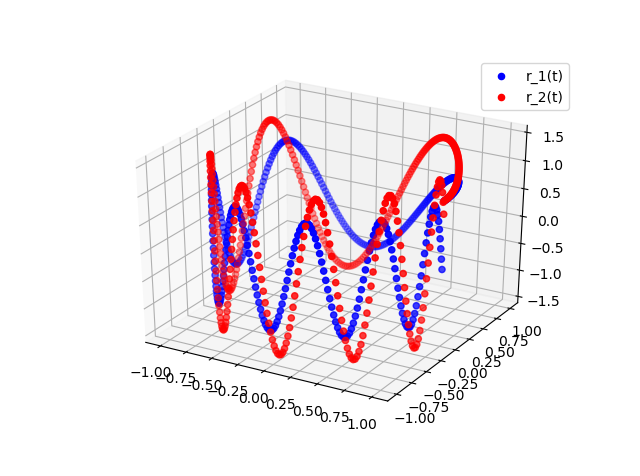}

\endminipage\hfill
\minipage{0.329\textwidth}
	 \includegraphics[width=\linewidth]{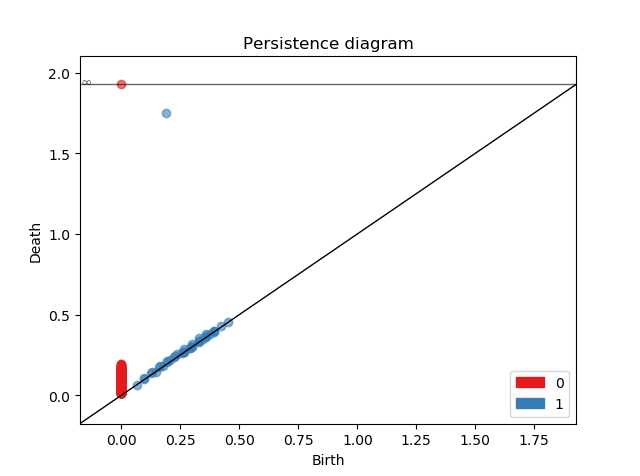}

\endminipage\hfill
\minipage{0.329\textwidth}
	 \includegraphics[width=\linewidth]{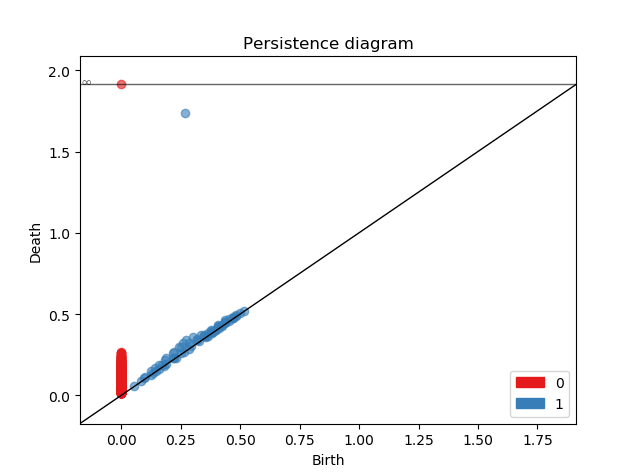}	

\endminipage\hfill

\caption{Graphs of $r_1(t)$ (in blue) and $r_2(t)$ (in red) (left) and their corresponding PDs (center and right, respectively).}\label{fig:cosslidingpics}
\end{figure}
pictured in Figure \ref{fig:cosslidingpics}. The Wasserstein distance between their corresponding persistent diagrams is approximately 0.987710397 and the Sliced Wasserstein distance between their persistent diagrams calculated with 50 directions is 0.47876553. Even though these two point clouds have extremely similar topology and similarly structured PDs, they have non-zero Wasserstein and Sliced Wasserstein distances between them. 

There are similar issues with previous Vectorization methods for PDs. Consider the vectorization of a PD in which the vector entries are simply the persistence of each point on the diagram, of a fixed homology dimension, sorted into decreasing order. Call this the {\it persistence vector}. Using the same example of the point clouds generated by $r_1(t)$ and $r_2(t)$ the largest four entries of the corresponding persistence vectors of dimension 1 (normalized so that the entries lie between 0 and 1 by dividing each entry by the maximum value) are
$$
\langle 1, 4.87754541\times 10^{-2}, 4.09349618\times 10^{-2}, 3.75080674\times 10^{-2} \rangle
$$
and
$$
\langle 1, 1.46535480\times 10^{-2}, 1.41529871\times 10^{-2}, 1.25857935 \times 10^{-2} \rangle
$$
respectively. The distance between these persistence vectors is nonzero (although the concept of closeness should be defined first) even though the topology of both point clouds is similar. 

This persistence vector has an additional issue with representing the structure of the point cloud faithfully. Consider a PD with four persistence points at $(0.01, 0.02)$ pictured in Figure \ref{fig:uniform} with a graph representing its underlying point cloud in 2D. In addition, consider another PD with four persistence points at $(0.01,0.02)$, $(0.02, 0.03)$, $(0.03, 0.04)$, and $(0.04,0.05)$ pictured in Figure \ref{fig:nonuniform} with a graph representing its underlying point cloud depicted in 2D. The distance between nodes in the figure is the Euclidean distance. Note that both of the corresponding point clouds exist in a high dimension and so in order to visualize these point clouds we draw them in 2D as graphs using TikZ. 
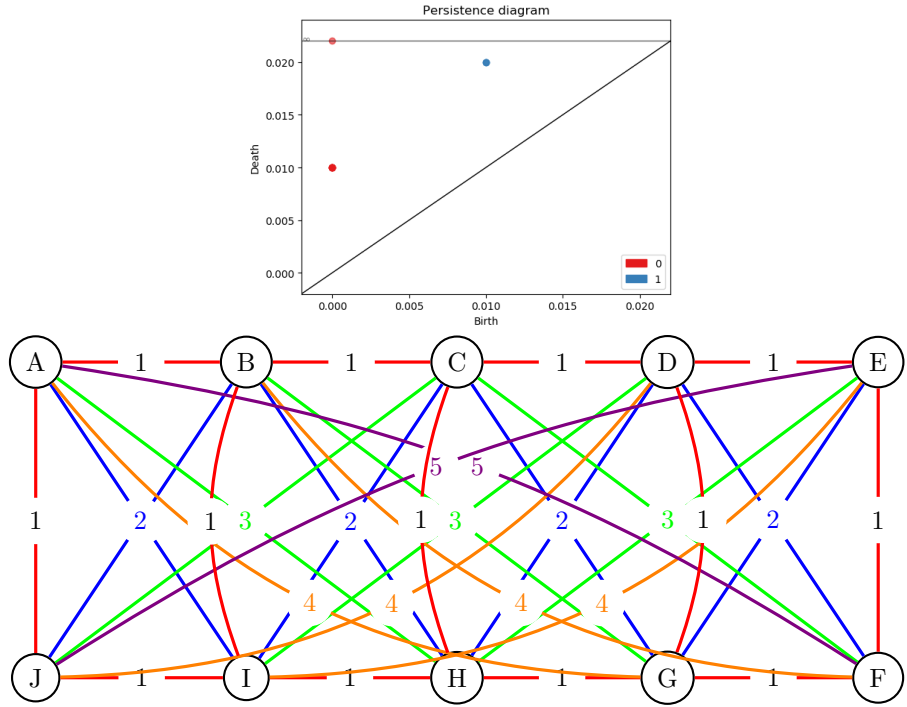
\begin{figure}[hbt!]
  \centering
\minipage{0.5\textwidth}
    \includegraphics[width=\linewidth]{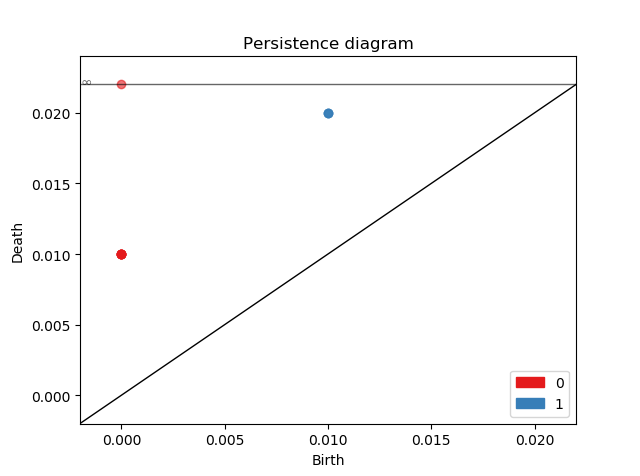}

\endminipage\hfill
\minipage{0.99\textwidth}
	\begin{tikzpicture}[scale=0.7]
\begin{scope}[every node/.style={circle,thick,draw}]
    \node (A) at (0,6) {A};
    \node (B) at (4,6) {B};
    \node (C) at (8,6) {C};
    \node (D) at (12,6) {D};
    \node (E) at (16,6) {E};
    \node (F) at (16,0) {F} ;
		\node (G) at (12,0) {G};
		\node (H) at (8,0) {H};
		\node (I) at (4,0) {I};
		\node (J) at (0,0) {J};
\end{scope}

\begin{scope}[>={Stealth[black]},
              every node/.style={fill=white,circle},
              every edge/.style={draw=red,very thick}]
    \path [-] (A) edge node {$1$} (B);
    \path [-] (B) edge node {$1$} (C);
	  \path [-] (C) edge node {$1$} (D);
	  \path [-] (D) edge node {$1$} (E);
	  \path [-] (E) edge node {$1$} (F);
	  \path [-] (F) edge node {$1$} (G);
	  \path [-] (G) edge node {$1$} (H);
	  \path [-] (H) edge node {$1$} (I);	  
	  \path [-] (I) edge node {$1$} (J);
	  \path [-] (J) edge node {$1$} (A);
	  \path [-] (A) edge[color=blue] node {$2$} (I);
	  \path [-] (B) edge[color=blue] node {$2$} (J);
	  \path [-] (B) edge[color=blue] node {$2$} (H);
	  \path [-] (C) edge[color=blue] node {$2$} (I);
	  \path [-] (C) edge[color=blue] node {$2$} (G);
	  \path [-] (D) edge[color=blue] node {$2$} (H);
	  \path [-] (D) edge[color=blue] node {$2$} (F);
	  \path [-] (E) edge[color=blue] node {$2$} (G);
	  \path [-] (A) edge[color=green] node {$3$} (H);
	  \path [-] (C) edge[color=green] node {$3$} (J);
	  \path [-] (B) edge[color=green] node {$3$} (G);
	  \path [-] (D) edge[color=green] node {$3$} (I);
	  \path [-] (C) edge[color=green] node {$3$} (F);
	  \path [-] (E) edge[color=green] node {$3$} (H);
	  \path [-] (A) edge[color=orange, bend right=25] node {$4$} (G);
	  \path [-] (B) edge[color=orange, bend right=25] node {$4$} (F);
	  \path [-] (J) edge[color=orange, bend right=25] node {$4$} (D);
	  \path [-] (I) edge[color=orange, bend right=25] node {$4$} (E);
	  \path [-] (A) edge[color=violet, bend left=12] node {$5$} (F);
	  \path [-] (J) edge[color=violet, bend left=12] node {$5$} (E);
	  \path [-] (B) edge[bend right=20] node {$1$} (I);
	  \path [-] (C) edge[bend right=20] node {$1$} (H);
	  \path [-] (D) edge[bend left=20] node {$1$} (G);
   
\end{scope}
\end{tikzpicture}

\endminipage\hfill
\caption{The persistence diagram with four of the same dimension 1 persistence points and a graph representation of its underlying point cloud. Nodes on the graph represent points in the point clouds and edge weights are the Euclidean distance between the nodes the edges connect. Edge weights on the graph are in tenths.}\label{fig:uniform}
\end{figure}

\begin{figure}[hbt!]
  \centering
\minipage{0.45\textwidth}
	 \includegraphics[width=\linewidth]{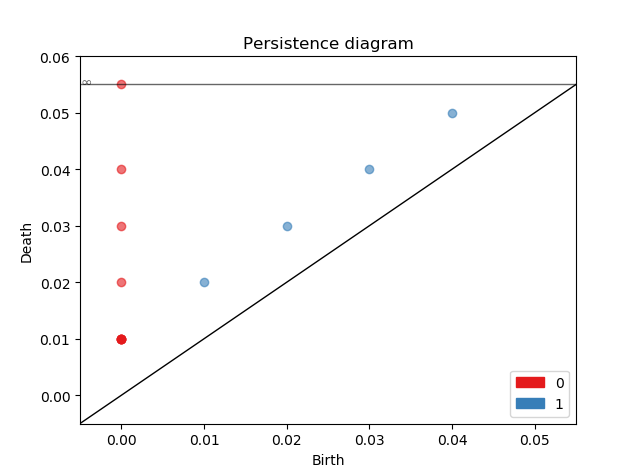}

\endminipage\hfill
\minipage{0.99\textwidth}
	 
\begin{tikzpicture}[scale=0.55]
\begin{scope}[every node/.style={circle,thick,draw}]
    \node (A) at (0,6) {A};
    \node (B) at (3,6) {B};
    \node (C) at (7,6) {C};
    \node (D) at (13,6) {D};
    \node (E) at (21,6) {E};
    \node (F) at (21,0) {F} ;
		\node (G) at (13,0) {G};
		\node (H) at (7,0) {H};
		\node (I) at (3,0) {I};
		\node (J) at (0,0) {J};
\end{scope}

\begin{scope}[>={Stealth[black]},
              every node/.style={fill=white,circle},
              every edge/.style={draw=red,very thick}]
    \path [-] (A) edge[color=red] node {$1$} (B);
    \path [-] (B) edge[color=blue] node {$2$} (C);
    \path [-] (C) edge[color=green] node {$3$} (D);
    \path [-] (D) edge[color=orange] node {$4$} (E);
    \path [-] (E) edge[color=red] node {$1$} (F);
    \path [-] (F) edge[color=orange] node {$4$} (G);
    \path [-] (G) edge[color=green] node {$3$} (H);
    \path [-] (H) edge[color=blue] node {$2$} (I);
    \path [-] (I) edge[color=red] node {$1$} (J);
    \path [-] (J) edge[color=red] node {$1$} (A);
    \path [-] (A) edge[color=blue] node {$2$} (I);
    \path [-] (A) edge[color=orange] node {$4$} (H);
    \path [-] (C) edge[color=orange] node {$4$} (J);
    \path [-] (A) edge[color=yellow] node {$8$} (G);
    \path [-] (D) edge[color=yellow] node {$8$} (J);
    \path [-] (B) edge[color=blue] node {$2$} (J);
    \path [-] (B) edge[color=red, bend right=17] node {$1$} (I);
    \path [-] (C) edge[color=red, bend right=40] node {$1$} (H);
    \path [-] (B) edge[color=green] node {$3$} (H);
    \path [-] (C) edge[color=green] node {$3$} (I);
    \path [-] (B) edge[color=cyan] node {$6$} (G);
    \path [-] (D) edge[color=cyan] node {$6$} (I);
    \path [-] (B) edge[color=magenta] node {$10$} (F);
    \path [-] (E) edge[color=magenta] node {$10$} (I);
    \path [-] (C) edge[color=orange] node {$4$} (G);
    \path [-] (D) edge[color=orange] node {$4$} (H);
    \path [-] (D) edge[color=violet] node {$5$} (F);
    \path [-] (G) edge[color=violet] node {$5$} (E);
    \path [-] (A) edge[color=pink, bend left=13] node {$11$} (F);
    \path [-] (E) edge[color=pink, bend left=13] node {$11$} (J);
    \path [-] (D) edge[color=red, bend left=55] node {$1$} (G);
    \path [-] (C) edge[color=yellow] node {$8$} (F);
    \path [-] (E) edge[color=yellow] node {$8$} (H);
   
\end{scope}
\end{tikzpicture}

\endminipage\hfill

\caption{The persistence diagram with four distinct dimension 1 persistence points and a graph representation of its underlying point cloud. Nodes on the graph represent points in the point clouds and edge weights are the Euclidean distance between the nodes the edges  connect. Edge weights on the graph are in tenths}\label{fig:nonuniform}
\end{figure}
Their corresponding persistence vectors in dimension 1 are exactly the same: 
$$ \langle 0.1, 0.1, 0.1, 0.1 \rangle $$
and yet the structure of each is significantly different. The first point cloud is quite uniformly spread out whereas the second point cloud is non-uniformly spread out. From the viewpoint of PH, the first data set's simplicial complex at filtration $\tau = 0.1$ has four loops and the other data set has four simplicial complexes in a row at $\tau = 0.1, 0.2,0.3,$ and $0.4$ with exactly one loop. 

Our aim with the interconnectivity vector is to eliminate these two issues of scale and topology seen in Wasserstein kernel methods and the persistence vectorization method. The interconnectivity vector will be a faithful representation of topology as it will be invariant to scale and it will more accurately represent the concurrence of homological features such as connected components, loops, vacuous volumes, etc. In addition, the interconnectivity vector is designed to be more computationally efficient than the Wasserstein kernel methods.

\section{The Interconnectivity Vector}

We will now introduce the interconnectivity vector. The interconnectivity vector is formed by counting the number of persistence points lying within a circle of radius equal to the persistence of each individual point. Note that a similar measure was presented in \cite{chung2009} where a persistence surface was generated by convolving a PD with the characteristic function of a disk of radius 0.2 around a fixed point in the diagram. Their surface, in effect, calculated the concentration of persistence points in that singular fixed disk. Our work can be seen as a specification and extension of this idea to every point in the diagram with topologically significant radii. The interconnectivity vector, rather than computing the concentration of persistence points in a specific area of the PD, will compute the extent to which the homological features (connected components, loops, vacuous volumes, etc.) coexist at every filtration level.

\subsection{Definition of the Interconnectivity Vector}
Let $B$ be a persistence diagram with finitely many, $N$, off-diagonal points $x_i = (b_i, d_i)$ viewed as a subset of $\mathbb{R}^2$, possibly with overlapping points. Let $d_{\Delta}(x_i)$ be the distance from $x_i$ to $\Delta$, the diagonal line $y = x$, using the $L_{\infty}$-norm. Then $d_{\Delta}(x_i) = || x_i - \Delta||_{\infty} = d_i - b_i$. Note that $d_i$ can not be smaller than $b_i$. 

Let $D_i$ be the disk defined by
$$ D_i = \{ (x,y) \in \mathbb{R}^2 : (x-b_i)^2 + (y-d_i)^2 < d_{\Delta}(x_i)^2\}$$
Every point off the diagonal which is in the intersection of $B$ and $D_i$ of $x_i$ corresponds to a homological feature that exists for at least some of the same filtration values as the feature corresponding to point $x_i$
\begin{figure}[hbt!]
  \centering
    \includegraphics[width=0.5\textwidth]{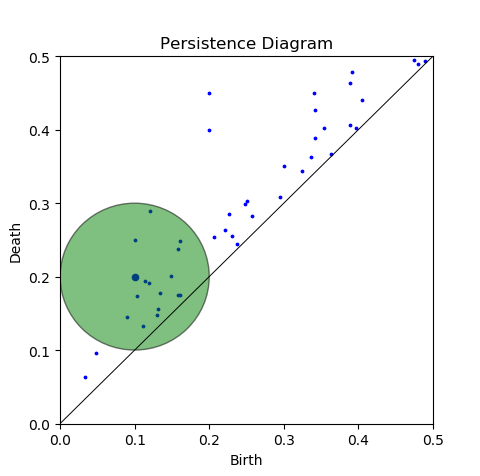}
\caption{A synthesized PD of fixed dimension on which a sample disk $D_i$ corresponding to the persistence point $x_i$ is drawn in green. Note that the radius of the disk $D_i$ is exactly the persistence of the point on the diagram that lies at the disk's center. Multiple persistence points exist inside $D_i$. }\label{fig:interpd}
\end{figure}
(see Figure \ref{fig:interpd}). Figure  \ref{fig:interpd} shows the schematic illustration of $D_i$ of $x_i$ in green and those persistence points that exist inside $D_i$. There are multiple persistence points inside $D_i$ for this example. On a barcode, these would be bars that overlap for some filtration values. 
For example, if $x_j$ lies in $D_i$ then $$(b_j - b_i)^2 + (d_j - d_i)^2 < d_{\Delta}(x_i)^2$$
But this implies that if $b_j > b_i$ (a similar argument holds for $b_j < b_i$)
$$b_j - b_i < d_{\Delta}(x_i)$$
$$ \implies b_j < b_i + d_{\Delta}(x_i)$$ But $b_i + d_{\Delta}(x_i) = b_i + (d_i - b_i) = d_i$ and so $b_j < d_i$. Thus the intersection of the persistence intervals $[b_i,d_i]$ and $[b_j, d_j]$ is nonempty. Note, however, that the reverse statement need not be true. The two persistence intervals $[b_i,d_i]$ and $[b_j,d_j]$ can overlap but $x_j$ need not lie in $D_i$. For example, consider the barcode 
\begin{figure}[hbt!]
  \centering
    \includegraphics[width=0.45\textwidth]{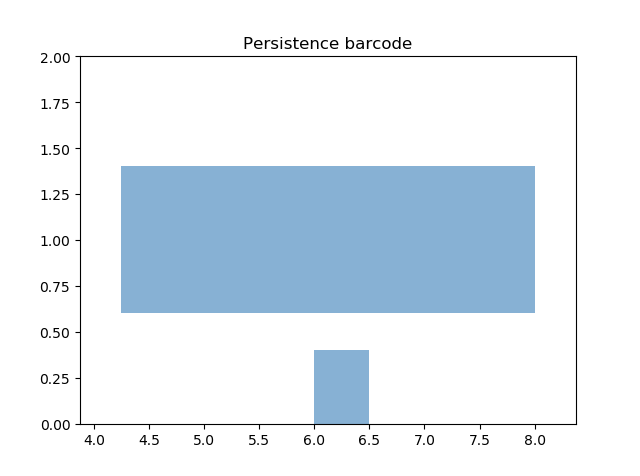}
\caption{A persistence barcode with overlapping persistence intervals but for which one point is not in the other point's disk.}\label{fig:barcodeoverlap}
\end{figure}
depicted in Figure \ref{fig:barcodeoverlap} with persistence intervals $[b_i, d_i] = [4.5,8]$ and $[b_j, d_j] = [6,6.5]$. Note that on the corresponding PD the distance between the points $(4.5,8)$ and $(6,6.5)$ using the Euclidean distance is $1.5\sqrt{2}$ which is greater than $0.5~  (= d_j - b_j)$ but less than $3.5~ (= d_i - b_i)$. Therefore the point $(6,6.5)$ is included in the disk of radius 3.5 centered at $(4.5,8)$ but the point $(4.5,8)$ is not in the disk of radius 0.5 centered at $(6,6.5)$. This is because the percentage of $[4.5,8]$ which is overlapped by $[6,6.5]$ is small in comparison to the lifetime of $[4.5,8]$. 

\begin{definition}
Let $B$ be a persistence diagram with $N$ off-diagonal points. Define the interconnectivity vector $\vec{v}$ to be the vector of length $N$ whose $i$th entry $v_i$ is the cardinality of $B \cap D_i$ as defined above.
\end{definition} 

\begin{remark} Note that $v_i \geq 1$ for all $i \in \{1, ... , N\}$. 
\end{remark}

For computations later on, we permute the entries of this vector so that $v_i \geq v_{i+1}$ for all $i = 1, \cdots, N-1$. Note that this vector can be recovered directly from the barcode. For each entry $v_i$ corresponding to the bar beginning at $b_i$ and ending at $d_i$ one needs to only count the number of bars whose birth and death satisfy the following condition
$$(b_i - birth)^2 + (d_i - death)^2 < (d_i - b_i)^2$$

The following two examples show a simple demonstration of the interconnectivity vector. 
\begin{example} \label{example:sliding} Consider the sliding window embeddings of $f(x) = \cos(x)$ and $g(x) = 3\cos(x)$, $SW_{M,\tau}: f(x)\in \mathbb{R} \rightarrow f^{sw} \in \mathbb{R}^{M+1}$, 
$ f^{sw} = (f(x), f(x+\tau), \cdots, f(t+M\tau))^T$ (same for $g(x)$),  generated using 100 points with window size $M\tau$ where $M = 2$ and $\tau = 6$ \cite{perea2013sliding}. Figure \ref{fig:slidingpics} contains their respective point clouds, red for $f(x)$ and blue for $g(x)$, respectively (left figure) and PDs for $f(x)$ (middle) and $g(x)$ (right). Note that in PDs, those red and blue points correspond to the zero- and one-dimensional homology, respectively. 
\begin{figure}[hbt!]
  \centering
\minipage{0.329\textwidth}
    \includegraphics[width=\linewidth]{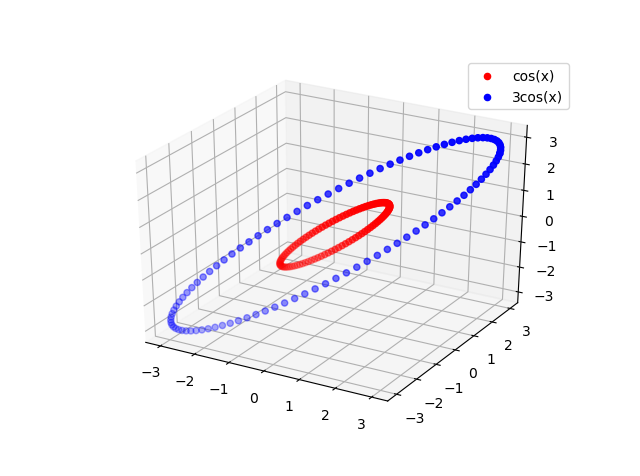}

\endminipage\hfill
\minipage{0.329\textwidth}
	 \includegraphics[width=\linewidth]{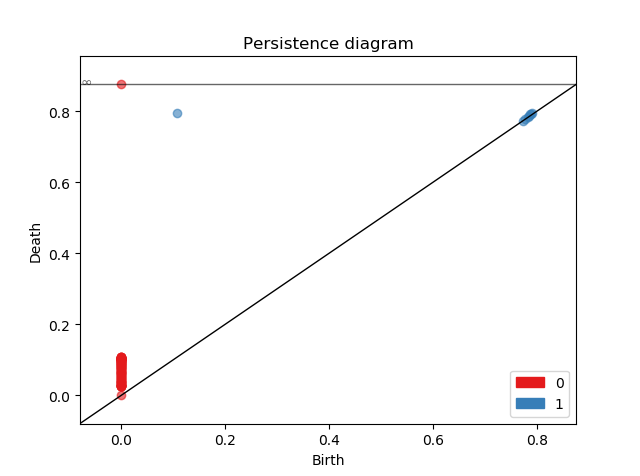}

\endminipage\hfill
\minipage{0.329\textwidth}
	 \includegraphics[width=\linewidth]{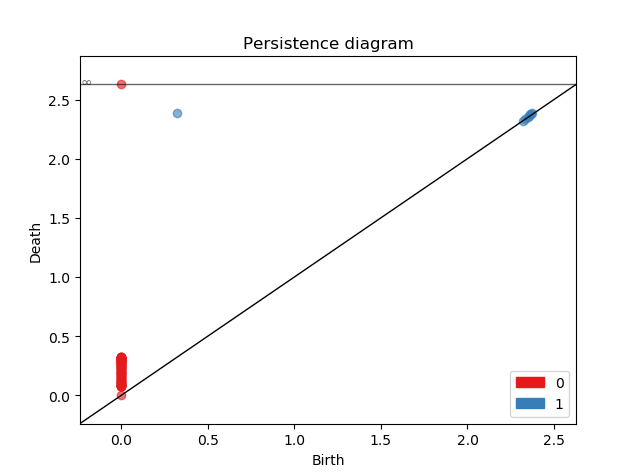}	

\endminipage\hfill

\caption{Left: Sliding window embeddings of $f(x) = \cos(x)$ (red) and $g(x) = 3\cos(x)$ (blue) with parameters $M = 2$ and $\tau = 6$. 
Middle: The persistence diagram for $f(x)$. Right: The persistence diagram for 
$g(x)$. 
}\label{fig:slidingpics}
\end{figure}
The PDs of each point cloud contain eight points in dimension $1$ and so the interconnectivity vector will be a vector of length eight.  For {\it both} point clouds the interconnectivity vector in dimension 1 is
$$\vec{v} = \langle 8, 3, 2, 2, 2, 1, 1, 1\rangle$$
Here the interconnectivity vector correctly indicates that the topology of the two point clouds are the same regardless of the scale, which, however, is not true of the topological vector described in the previous section. 
\end{example}

\begin{example}\label{random3D} Consider a uniform random 3D point cloud pictured in Figure \ref{fig:rndpics} with 100 points. 
\begin{figure}[hbt!]
  \centering
    \includegraphics[width=0.45\textwidth]{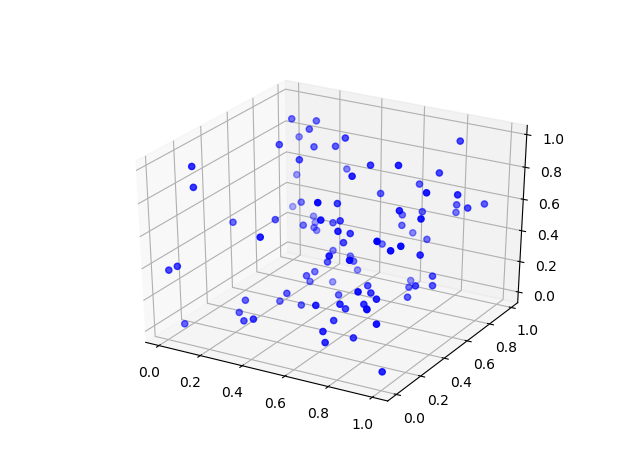}
	 \includegraphics[width=0.45\textwidth]{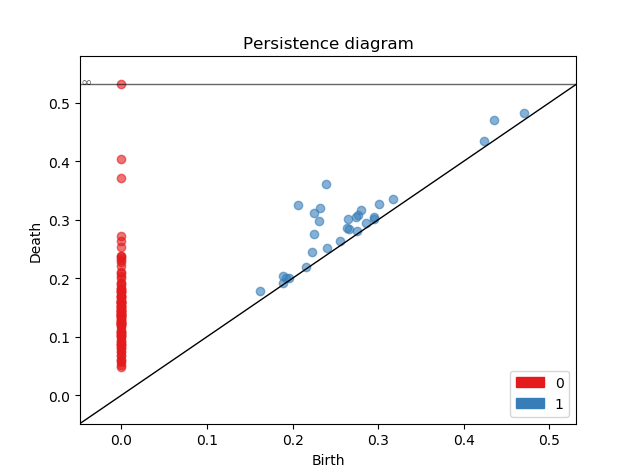}
\caption{Left, a uniform random point cloud in 3-space generated by Numpy's random.rand() command. Right, the corresponding persistence diagram generated by GUDHI.}\label{fig:rndpics}
\end{figure}
Again, working with the 1-dimensional homology, the corresponding PD contains 30 points and so the interconnectivity vector is the following vector of length 30:
$$ \vec{v}(B) = \langle 22, 22, 17, 17,  7,  5,  4,  3,  3,  3,  3,  2,  1, 1,  1,  1,  1,  1,  1,  1,  1,  1,  1,  1,  1,  1, 1,  1,  1,  1\rangle$$

\end{example}

\subsection{Instability}

The interconnectivity vector, as described above, is not stable with respect to the Wasserstein-1 metric as a slight change in a birth-death pair could result in a large change in the vector. Recall the definition of the $p$-Wasserstein distance between persistence diagrams \cite{wasserstein:1969}.

\begin{definition}\label{wasserstein} The $p$-Wasserstein distance between persistence diagrams $X$ and $Y$ is
$$W_p(X,Y) = \left[ \inf_{\eta:X \to Y} \sum_{x \in X} || x - \eta(x)||_{\infty}^p \right]^{1/p}$$
where $\eta:X \to Y$ is a partial matching of $X$ and $Y$. 
\end{definition}
Note that as $p \to \infty$ this distance becomes the {\it bottleneck distance}. 

\begin{thm}
Let $B$ and $B'$ be persistence diagrams containing only finitely many off-diagonal points. The interconnectivity vector is \textbf{not stable} with respect to the 1-Wasserstein distance. 
\end{thm}

\begin{proof} We prove by example. For simplicity, assume the persistence diagram $B$ contains only two points off the diagonal: $x_1 = (1,2)$ and $x_2 = (b,d)$. For any $0 < \epsilon < 2$ let 
$$b = \sqrt{2}\left(\frac{1}{2} - \frac{\epsilon}{4}\right) + 1, \quad d = \sqrt{2}\left(\frac{1}{2} - \frac{\epsilon}{4}\right) + 2$$
Further, let $D_1$ and $D_2$ be the following disks 
$$D_1 = \{ (x,y) \in \mathbb{R}^2 : (x - 1)^2 + (y-2)^2 < 1\}$$ 
$$D_2 = \{ (x,y) \in \mathbb{R}^2 : (x - b)^2 + (y-d)^2 < 1\}$$

Then the Euclidean distance between $x_1$ and $x_2$ is 
$$ \text{dist}(x_1,x_2) = 1 - \frac{\epsilon}{2} < 1$$
Since the radii of $D_1$ and $D_2$ are equal and $\text{dist}(x_1,x_2) < 1$ the interconnectivity vector for the persistence diagram $B$ is
$$\vec{v}(B) = \langle 2,2 \rangle$$

\bigskip

Now consider $B'$, the persistence diagram with two points $y_1 = (1,2)$ and $y_2 = (b,d')$ where $b$ and $d$ are the same as in the case of $B$ above and $$d' = d + \frac{\sqrt{2}\epsilon}{2}= \sqrt{2}\left(\frac{1}{2} + \frac{\epsilon}{4}\right) + 2$$
Further, let $D'_1$ and $D'_2$ be the following disks 
$$D'_1 = \Bigg\{ (x,y) \in \mathbb{R}^2 : (x - 1)^2 + (y-2)^2 < 1\Bigg\}$$ 
$$D'_2 = \Bigg\{ (x,y) \in \mathbb{R}^2 : (x - b)^2 + (y-d')^2 < \left(1 + \frac{\sqrt{2}\epsilon}{2}\right)^2\Bigg\}$$
The Euclidean distance between $y_1 = (1,2)$ and $y_2 = (b,d')$ is for $\epsilon>0$
$$ \text{dist}(y_1,y_2)  =  \sqrt{1 + \frac{\epsilon^2}{4}} <1 + \frac{\epsilon}{\sqrt{2}}$$

Note that for $0 < \epsilon < 2$, the distance between $y_1$ and $y_2$ is less than the persistence $d'-b = 1 + \frac{\sqrt{2}\epsilon}{2}$. Since the radius of $D'_1 = 1$ and the radius of $D'_2 = 1 + \frac{\sqrt{2}\epsilon}{2}$ but $\text{dist}(y_1,y_2)  <  1 + \frac{\sqrt{2}\epsilon}{2}$ the interconnectivity vector for the persistence diagram $B'$ is 
$$\vec{v}(B') = \langle 2,1 \rangle$$

In order to be stable with respect to the 1-Wasserstein distance, the interconnectivity vector should fulfill the following condition: 
$$ || \vec{v}(B) - \vec{v}(B') ||_{\infty} \leq C \cdot W_1(B,B') $$
for some non-negative constant $C$.
The $W_1(B,B')$ (Definition \ref{wasserstein}) is given by
$$W_1(B,B') = \inf_{\eta:B \to B'} \sum_{x \in B} ||x - \eta(x)||_{\infty}$$
where $\eta$ is a partial matching of $B$ and $B'$ and so 
\begin{align*} 
|| \vec{v}(B) - \vec{v}(B') ||_{\infty} &\leq C \cdot W_1(B,B') \\
&= C ||(b,d) - (b,d')||_{\infty}\\
&= C\frac{\sqrt{2}\epsilon}{2}
\end{align*} 
Since $ || \vec{v}(B) - \vec{v}(B') ||_{\infty}  = ||\langle 2,2 \rangle - \langle 2,1 \rangle ||_{\infty} = 1$ we have that $C \geq \frac{\sqrt{2}}{\epsilon}$. However, as $\epsilon > 0$ can be arbitrarily small, there does not exist such a constant C. Hence, the interconnectivity vector is not stable with respect to the 1-Wasserstein distance. 
\end{proof}

\section{The Stable Interconnectivity Vector}

In this section we will retool the definition of the interconnectivity vector to produce a version which is stable under the 1-Wasserstein distance. The general idea behind the stabilization is to do away with the duality of a persistence point being in or not in the disk $D_i$ and in its place we will weight each point inversely proportional to is distance from each persistence point. This stabilization is similar to those of \cite{zielinski2018persistence, Pachauri2011}. 

In \cite{zielinski2018persistence} a clustering algorithm (K-Means or Gaussian Mixture Model) is used to divide a PD into a desired number of {\it codebooks} in which each persistence point is counted to form a histogram vector. Their first approach is proven unstable and they apply a series of Gaussian distributions (derived from the Gaussian Mixture Model) on each codebook. The authors create a new stabilized {\it Persistence Bag of Words} vector by calculating the weighted average probability that persistence point $x_i$ was generated by each Gaussian distribution. 

In \cite{Pachauri2011} a concentration map of a PD is generated on a uniform grid of size $[0,5]\times[0,5]$ and then smoothed with a Gaussian kernel of bandwith 0.2. A vector is then created by computing the PDF of the concentration map.

For our stabilized version of the interconnectivity vector we will use this Gaussian smoothing of the PD idea. In particular, we will center a Gaussian distribution at each off-diagonal persistence point in the diagram and force the covariance of the distribution to depend on the persistence of that persistence point. Our vector will then be calculated by averaging the probability that the persistence point $x_i$ is generated by each Gaussian centered at the points in the PD. In effect, we will be calculating the concentration of persistence points which are ``close" to each off-diagonal point on the PD. 

\subsection{Definition of the Stable Interconnectivity Vector}
Suppose $B$ is a persistence diagram with $N$ off-diagonal points $x_i = (b_i,d_i)$. Let $\{ G_i ; \mu_i, \Sigma_i\}_{i = 1}^{N}$ be a collection of Gaussian distributions \cite{zielinski2018persistence} where $\mu_i$, the mean of $G_i$, is equal to $x_i \in B$ and where $\Sigma_i$, the covariance matrix for $G_i$, is 
 \[\Sigma_i = \begin{pmatrix} 
d_i-b_i + \delta& 0 \\
0 & d_i-b_i + \delta
\end{pmatrix} \]
Here $\delta$ is a small regularization term which is useful in the case that $d_i - b_i$ is arbitrarily small (one could do away with $\delta$ if we instead sample the PD and keep only those points with significant persistence).
Define $w_i$ to be $$w_i = \frac{1}{N}$$
Note that each $w_i > 0$ and $\sum_{i = 1}^{N} w_i = 1$.

Then the stable interconnectivity vector is
$$ \vec{v}^s = \left( v_i^s = w_i  \sum_{j =1}^{N} p_i(x_j|\mu_i,\Sigma_i)\right)_{i = 1}^{N}$$
where $$p_i(x_j | \mu_i,\Sigma_i) = \frac{\exp\left(-\frac{1}{2}(x_j - \mu_i)^T \Sigma_i^{-1}(x_j - \mu_i)\right)}{2\pi |\Sigma_i|^{1/2}}$$
is the probability density function of Gaussain $G_i$ at $x_j$ \cite{zielinski2018persistence}. Note that by $|\Sigma_i|$ we mean the determinant of the covariance matrix. 

\begin{example} Returning to the sliding window embedding of $f(x) = \cos(x)$ in Example \ref{example:sliding}, recall that the persistence diagram in dimension 1 for this example had exactly eight off-diagonal points seven of which have persistence on the order of $10^{-3}$ or smaller and one point with persistence equal to $6.89 \times 10^{-1}$. The stable interconnectivity vector for this persistence diagram of dimension 1 using parameter value $\delta = 0.5$ is
%
\begin{eqnarray}
\vec{v}^s(B) =& \langle& 0.30386786, 0.30328569, 0.30269939, 0.30212169, 
\nonumber \\
& & 0.30156299, 0.30103144, 0.30053287, 0.11332091 \rangle \nonumber
\end{eqnarray}
\end{example}

\subsection{Stability}
The stable interconnectivity is so named as it is stable with respect to the 1-Wasserstein distance. This means that given a small change in the persistence diagram causes at most a small change in the vector. We now state this as a theorem and provide a proof.

\begin{thm}\label{thm:stability}
Let $B$ be a persistence diagram of finite size and $B'$ be the persistence diagram obtained by perturbing $B$ by an arbitrary $\epsilon > 0$ such that $W_1(B,B') \leq \epsilon$. Then there exists a non-negative constant $C < \infty$ for any $\epsilon$ such that  
$$||\vec{v}^s(B) - \vec{v}^s(B')||_{\infty} \leq C \cdot W_1(B,B')$$
\end{thm}

\begin{proof} Let $N$ be the number of off-diagonal points in $B$. Let $\eta: B \to B'$ be the partial matching that realizes the 1-Wasserstein distance between $B$ and $B'$. For a fixed $i \in \{1,...,N\}$. we have
\begin{align*} 
||v_i^s(B) - v_i^s(B')||_{\infty} &= ||w_i \sum_{j = 1}^{N} \left(p_i(x_j| \mu_i, \Sigma_i) - p_i(\eta(x_j) | \mu_i,\Sigma_i)\right)||_{\infty}\\
& \leq w_i \sum_{j = 1}^{N} ||p_i(x_j | \mu_i,\Sigma_i) - p_i(\eta(x_j) | \mu_i,\Sigma_i)||_{\infty}\\
\end{align*}
As $p_i: \mathbb{R}^2 \to \mathbb{R}$ is continuously differentiable it is also Lipschitz continuous with Lipschitz constant $L_i$. We get
\begin{eqnarray}
w_i \sum_{j = 1}^{N} ||p_i(x_j | \mu_i,\Sigma_i) - p_i(\eta(x_j) | \mu_i,\Sigma_i)||_{\infty} &\leq& w_i \sum_{j = 1}^{N} ||L_i(x_j - \eta(x_j))||_{\infty}\nonumber \\
 &= &w_i L_i \sum_{j = 1}^{N} ||x_j - \eta(x_j)||_{\infty}
 \nonumber \\
 &= & w_i L_i  W_1(B, B') \nonumber 
\end{eqnarray}
If we let $$C = \max_{i \in [1,...,N]} w_i L_i$$ we have the desired result.

\end{proof}

\begin{example}

Returning to the example used to show that the original interconnectivity vector was unstable, we will now show that Theorem \ref{thm:stability} is satisfied for this example.

Let $B$ be the persistence diagram that contains two off-diagonal points: $x_1 = (1,2)$ and $ x_2 = (b,d)$, where $b$ and $d$ are the same as before 
$$b = \sqrt{2}\left(\frac{1}{2} - \frac{\epsilon}{4}\right) + 1, 
\quad
 d = \sqrt{2}\left(\frac{1}{2} - \frac{\epsilon}{4}\right) + 2$$
for any $0 < \epsilon < 2$.
%
%
%
%
The stable interconnectivity vector for the persistence diagram $B$ is easily computed as the vector $\vec{v}^S(B) = \langle v^S_1, v^S_2 \rangle ^T$ where
$$ v^S_1 = v^S_2 =  \frac{1}{4 \pi (\delta+1)} \left[ 1 + \exp\left(-\frac{2}{\delta + 1} \left(\frac{1}{2} - \frac{\epsilon}{4}\right)^2 \right) \right]$$
Now let $B'$ be the persistence diagram that contains two off-diagonal points: $y_1 = (1,2)$ and $y_2 = (b,d')$, where 
$$b = \sqrt{2}\left(\frac{1}{2} - \frac{\epsilon}{4}\right) + 1, \quad 
d' = \sqrt{2}\left(\frac{1}{2} + \frac{\epsilon}{4}\right) + 2$$
for any $0 < \epsilon < 2$.
The stable interconnectivity vector for the persistence diagram $B'$ is also easily computed as the vector $\vec{v}^{S} (B') = \langle v^S_1, v^S_2 \rangle ^T$ where
\begin{align*} v^S_1 &= w_1 \left[ p_1(y_1 | \mu_1,\Sigma_1) + p_1(y_2| \mu_1,\Sigma_1) \right]\\
 &= \frac{1}{4\pi(\delta+1)} \left[ 1 +\exp \left(-\frac{1}{\delta +1} \left[ \left(\frac{1}{2} - \frac{\epsilon}{4}\right)^2 + \left(\frac{1}{2} + \frac{\epsilon}{4}\right)^2 \right] \right)\right]
\end{align*}
and 
\begin{align*} v^S_2 &= w_2 \left[ p_2(y_1| \mu_2,\Sigma_2) + p_2(y_2 | \mu_2,\Sigma_2) \right]\\
 &= \frac{\alpha}{2\pi} \left[ \exp \left(- 2\alpha\left[ \left(\frac{1}{2} - \frac{\epsilon}{4}\right)^2 + \left(\frac{1}{2} + \frac{\epsilon}{4}\right)^2 \right] \right)+ 1\right]
\end{align*}
where $\alpha = \frac{1}{2 + \sqrt{2}\epsilon + 2\delta}$.
\bigskip
Thus, we obtain 

$$\vec{v}^S(B) - \vec{v}^S(B')   =$$
\begingroup
\tiny
$$
\begin{bmatrix} \frac{1}{4\pi (\delta+1)} \left[ \exp\left( -\frac{2}{\delta+1} \left(\frac{1}{2} - \frac{\epsilon}{4}\right)^2\right) -  \exp \left(-\frac{1}{\delta+1} \left[ \left(\frac{1}{2} - \frac{\epsilon}{4}\right)^2 + \left(\frac{1}{2} + \frac{\epsilon}{4}\right)^2 \right] \right)\right] \\ \\
\frac{1}{2\pi} \left[  \frac{1-2\alpha(\delta+1)}{2(\delta + 1)} + \frac{1}{2(\delta+1) }\exp\left(-\frac{2}{\delta +1} \left(\frac{1}{2} - \frac{\epsilon}{4}\right)^2\right) - \alpha \cdot \exp\left( -2\alpha \left[ \left(\frac{1}{2} - \frac{\epsilon}{4}\right)^2 + \left(\frac{1}{2} + \frac{\epsilon}{4}\right)^2\right]\right) \right]
\end{bmatrix}
$$
\endgroup

Note that for all $0 < \epsilon < 2$, the absolute value of the first entry in the vector above, pictured in red in Figure \ref{fig:stableexample}, is less than or equal to that of the second entry in the vector, pictured in blue. 
\begin{figure}[hbt!]
  \centering
    \includegraphics[width=0.49\textwidth]{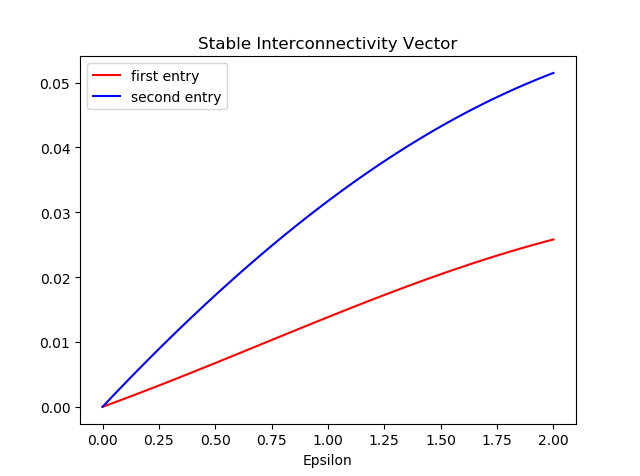}
	\includegraphics[width=0.49\textwidth]{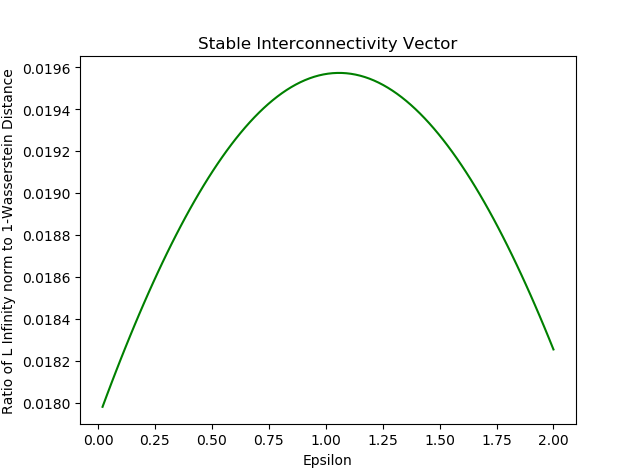}
\caption{Left: graph of the change in the stable interconnectivity vector versus the amount of perturbation $\epsilon$ using $\delta = 0.5$. The first vector entry is in red and the second is in blue. Right: graph of the ratio of the $L_{\infty}$-norm of the change in the stable interconnectivity vector to the 1-Wasserstein distance between PDs, $B$ and $B'$.}\label{fig:stableexample}
\end{figure}
Thus we obtain

$$|| \vec{v} - \vec{v} \, ' ||_{\infty} =$$
\begingroup
\tiny
$$ \frac{1}{2\pi} \Bigg| \frac{1-2\alpha(\delta+1)}{2(\delta + 1)} + \frac{1}{2(\delta+1) }\exp\left(-\frac{2}{\delta +1} \left(\frac{1}{2} - \frac{\epsilon}{4}\right)^2\right) - \alpha \cdot \exp\left( -2\alpha \left[ \left(\frac{1}{2} - \frac{\epsilon}{4}\right)^2 + \left(\frac{1}{2} + \frac{\epsilon}{4}\right)^2\right]\right) 
\Bigg|
$$
\endgroup

Recall, in order to satisfy Theorem \ref{thm:stability} we need
\begin{align*} ||\vec{v}^S(B) - \vec{v}^S(B')||_{\infty} &\leq C \cdot W_1(B, B')\\
 &= C ||(b,d) - (b, d')||_{\infty}\\ &= C \frac{\sqrt{2}\epsilon}{2}
\end{align*}
This means that we need to find a constant $C$ such that
$$ \frac{\sqrt{2}}{\epsilon}\cdot||\vec{v}^S(B) - \vec{v}^S(B')||_{\infty} \leq C$$
In Figure \ref{fig:stableexample} there is the graph the left hand side of the inequality above versus $\epsilon$. The graph attains a maximum value of approximately $0.01957236318939473$ when $\epsilon$ equals $1.\bar{03}$.
Thus there exists $C$, say,  $C = 0.02$, so that 
$$||\vec{v}^S(B) - \vec{v}^S(B')||_{\infty} \leq C \frac{\sqrt{2}\epsilon}{2}$$
and Theorem \ref{thm:stability} is satisfied.

\end{example}

Consider Example \ref{random3D} where we have 150 uniformly randomly selected points in the unit cube. Now perturb one randomly chosen point by a small amount $\epsilon$ and let us compare the effects of this small perturbation on the $L_{\infty}$ norm of the difference in each version of the interconnectivity vector. 
\begin{figure}[hbt!]
  \centering
    \includegraphics[width=0.45\textwidth]{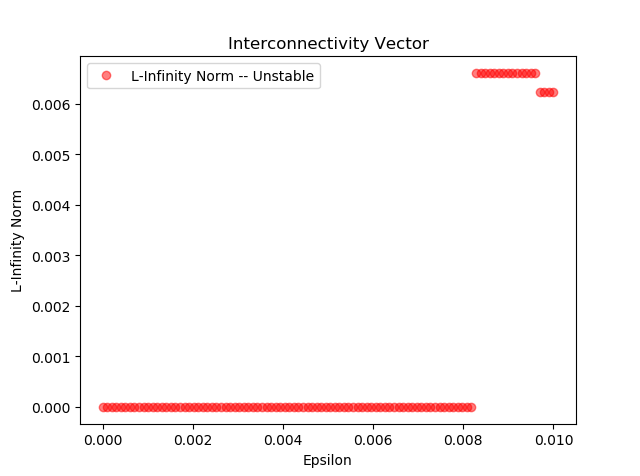}
	 \includegraphics[width=0.45\textwidth]{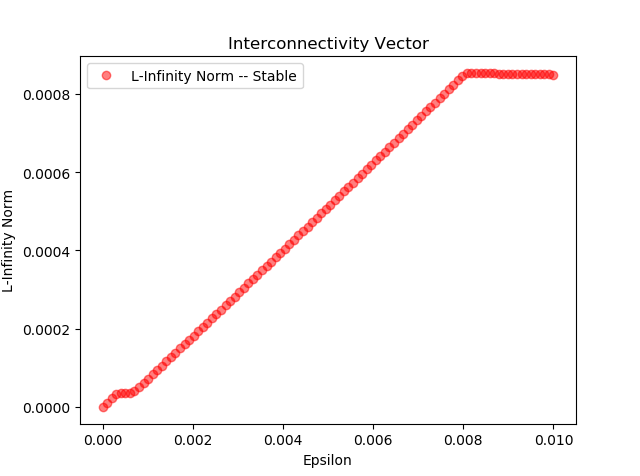}
    \includegraphics[width=0.45\textwidth]{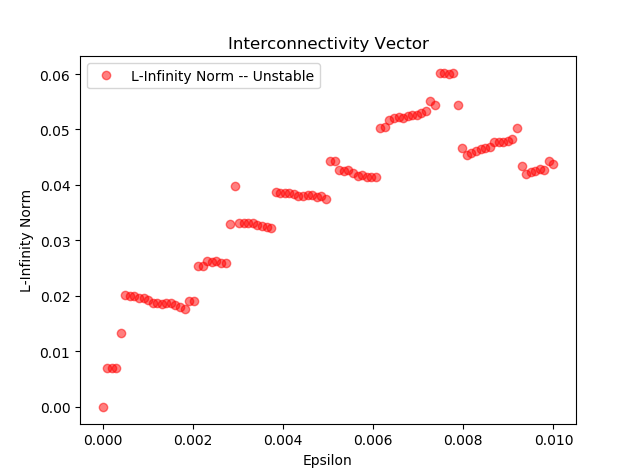}
	 \includegraphics[width=0.45\textwidth]{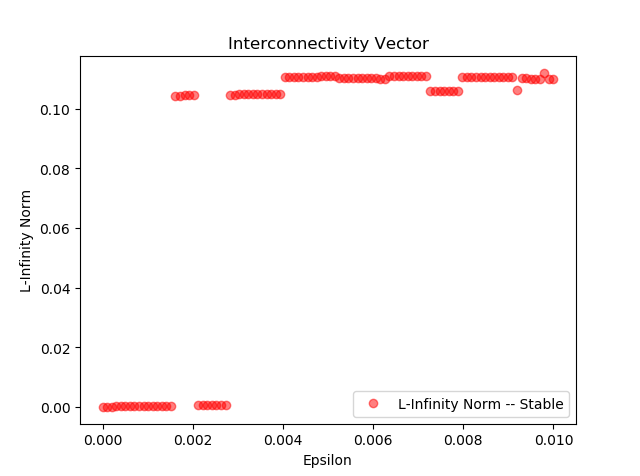}
\caption{Unstable interconnectivity vector (left) and stable interconnectivity vector (right). One point changed (upper) and randomly many points changed (lower).}\label{fig:Linfpics}
\end{figure}
In the top row of Figure \ref{fig:Linfpics} we can see the differences (vertical axis) in the two different vectors' response to a small perturbation $\epsilon$ (horizontal axis) of a single data point. The figure on the left is the change in the unstable interconnectivity vector as $\epsilon$, the amount by which the single point in the cloud is perturbed, changes. In this graph, we can see that when $\epsilon$ is very small (less than $0.008$) the point is not changed enough to affect its corresponding PD. Most likely the persistence points corresponding to this perturbed area of the point cloud are either not moved or are not moved enough to change disks $D_i$. Note, however, that when $\epsilon$ is larger than $0.008$ suddenly there is a large, discontinuous change of about $0.0065$ in the unstable interconnectivity vector. This jump is not present for the stable vector pictured in the top right of Figure \ref{fig:Linfpics} -- a fact which was assured by Theorem \ref{thm:stability}. Looking at the graph of the change in the stable interconnectivity vector we can see a continuous change followed by a plateau which most likely again corresponds to a period where the corresponding persistence points are either not moved or are moved so little that the change in probability $p_i$ is small.

In the bottom row of Figure \ref{fig:Linfpics} we can see the effect of perturbing randomly many points in the data set by a small amount $\epsilon$. In this experiment we used Numpy random.randint command to choose uniformly randomly many of the $150\times3$ coordinates in the point cloud to perturb by $\epsilon$. It appears that, when many small changes are compounded, the unstable interconnectivity vector appears to be more resistant to large jumps in the infinity norm than the stable interconnectivity vector. Here note that the point cloud is not changing continuously with $\epsilon$ but for each $\epsilon$ different points are selected (randomly) from the original point cloud and perturbed by $\epsilon$. Thus we observe more jumps in the left figure of the unstable version of the interconnectivity vector. There are discontinuous jumps in the $L_{\infty}$ norm of the stable interconnectivity vector but this does not contradict Theorem \ref{thm:stability} as the 1-Wasserstein distance between the original PD and the perturbed PD is larger than $\epsilon$.


\section{Numerical Examples}

We will now apply both the stable and unstable interconnectivity vectors to three examples. The first example involves four differently structured random and ordered point clouds. The second example will use a linked-twist map to generate orbits for four distinct parameter values. The final example involves a basic approach to optical character recognition of the Roman alphabet.

\subsection{Random vs. Uniform vs. Fractal Example}
\begin{figure}[hbt!]
  \centering
\minipage{0.24\textwidth}
    \includegraphics[width=\linewidth]{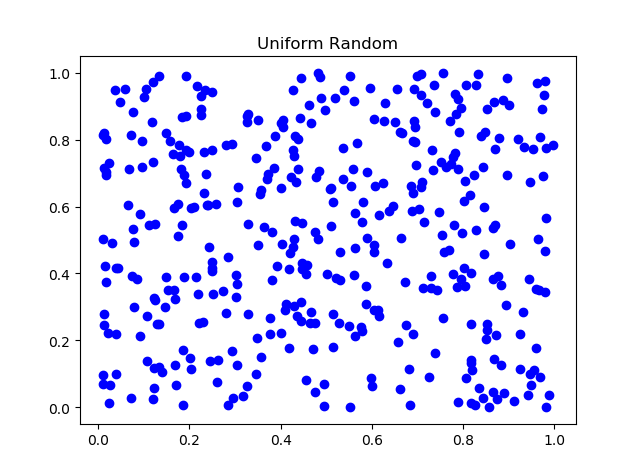}

\endminipage\hfill
\minipage{0.24\textwidth}
	 \includegraphics[width=\linewidth]{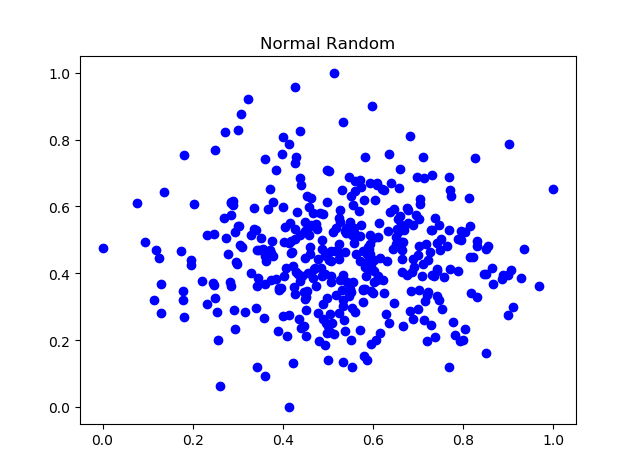}

\endminipage\hfill
\minipage{0.24\textwidth}
	 \includegraphics[width=\linewidth]{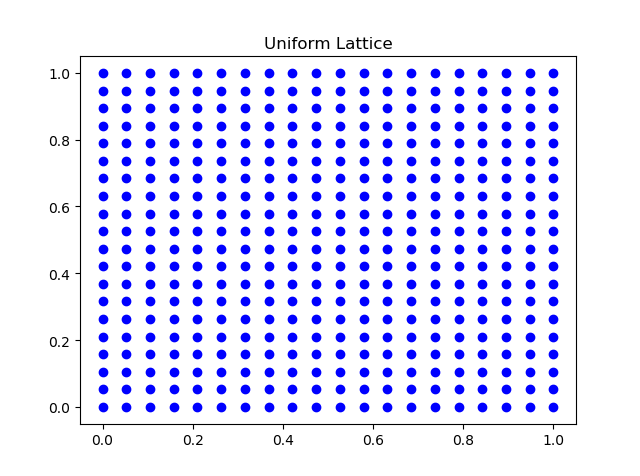}

\endminipage\hfill
\minipage{0.24\textwidth}
	 \includegraphics[width=\linewidth]{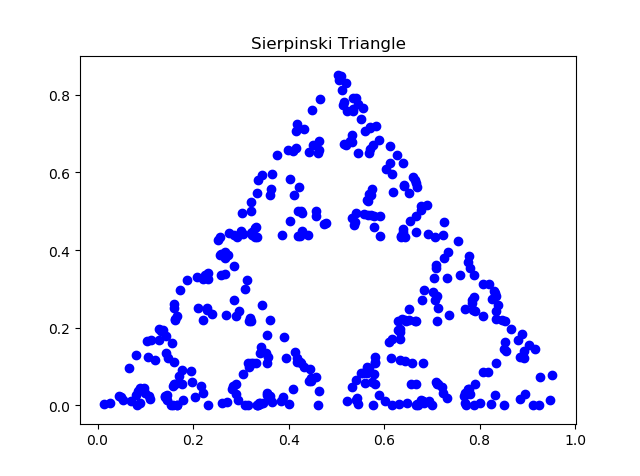}

\endminipage\hfill

\caption{Left to right: a uniform random point cloud, a normal random cloud with center (0.5,0.5) and variance 1 normalized to lie in $[0,1]^2$, a uniform lattice of points, and points from the Sierpinski triangle. Each point cloud contains 400 points.}\label{fig:randomclouds}
\end{figure}
For this experiment we considered four types of point clouds: a uniformly random distribution on $[0,1] \times [0,1]$, a normal random distribution with center at $(0.5,0.5)$ and variance 1 normalized to lie in $[0,1] \times [0,1]$, a uniform lattice of points on $[0,1] \times [0,1]$, and points drawn from the Sierpinski triangle created by the {\it chaos game} \cite{barnsley:1988}.
A sample point cloud for each type is shown in Figure \ref{fig:randomclouds}. Each point cloud has a total of $400$ points. 

Each point cloud was generated ten times and each time the unstable interconnectivity vector, the stable interconnectivity vector, and the persistence vector were calculated from the corresponding 1-dimensional PDs. A value of $\delta = 0.5$ was used in the computation of the stable interconnectivity vector. For each type of point cloud the vectors were averaged over the ten iterations. 
\begin{figure}[hbt!]
  \centering
\minipage{0.6\textwidth}
    \includegraphics[width=\linewidth]{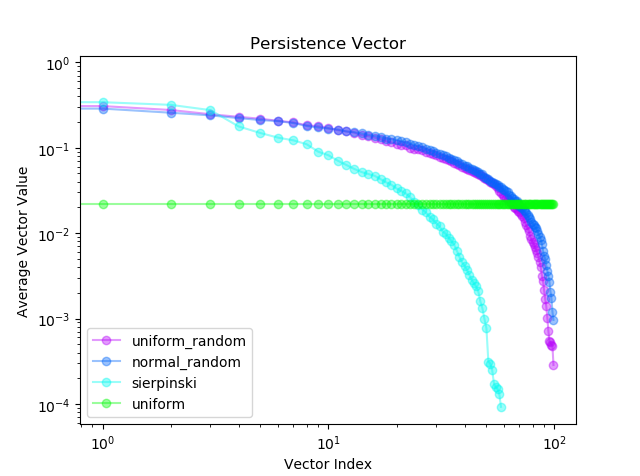}

\endminipage\hfill
\minipage{0.6\textwidth}
	 \includegraphics[width=\linewidth]{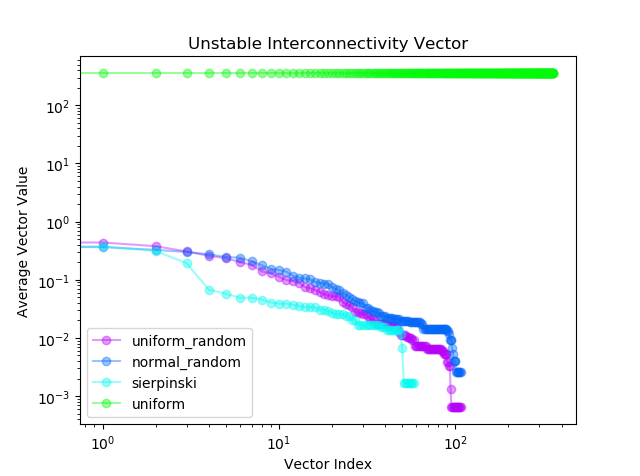}

\endminipage\hfill
\minipage{0.6\textwidth}
	 \includegraphics[width=\linewidth]{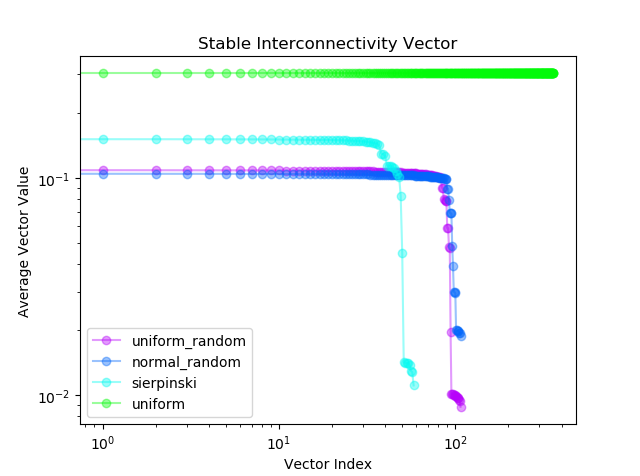}

\endminipage\hfill

\caption{Average vector values for each type of point cloud versus vector index in log-log scale beginning with the persistence vector, followed by the unstable interconnectivity vector, ending with the stable interconnectivity vector.}\label{fig:randomvectors}
\end{figure}
In Figure \ref{fig:randomvectors} there are plots of the average vector value versus the vector index for each of the three vectors in log-log scale. In each of the three graphs there is clear distinction between the uniform lattice point cloud (in green) and the other three types of point clouds. These plots show that both the unstable and stable interconnectivity vectors provide more clear distinction than the persistence vector, particularly for the lattice point cloud. Between the unstable and stable interconnectivity vectors, the stable interconnectivity vector yields a more distinctive power -- there is a large distinction in the stable interconnectivity vector values between the uniform and normal random and the Sierpinski stable interconnectivity vectors, particularly in their low indices.  This is as to be expected because although the Sierpinski cloud is not as evenly spaced as the lattice, it is not a random point cloud but rather it has a defined structure.  

\begin{figure}[h]
  \centering
    \includegraphics[width=0.5\textwidth]{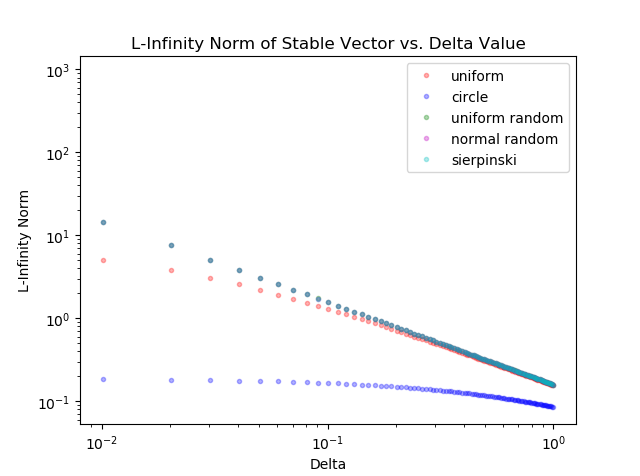}
\caption{The $L_\infty$ norm of the stable interconnectivity vector versus the regularization parameter $\delta$ in log-log scale for the lattice, uniform circle, uniform random and normal random in square and Sierpinski point clouds.}\label{fig:delta}
\end{figure}

As we conclude the demonstration of the first numerical example, we provide a brief discussion on the effect of the regularization term $\delta$ used for the regularization of the covariance vector when the birth and death values are too close to each other. 

Figure \ref{fig:delta} shows the $L_\infty$ norm of the stable interconnectivity vector versus the regularization parameter $\delta$ in log-log scale for the lattice, uniform circle, uniform random, normal random, and Sierpinski point clouds each with 400 points in the unit square. This figure shows the dominating effect of $\delta$ when its value is large. Thus, this graph suggests that one should choose smaller values of $\delta$ in applications of the stable interconnectivity vector.

%
%


\subsection{Linked-Twist Map Example} In this example we considered the linked-twist map which is a discrete dynamical system modeling fluid flow used in \cite{hertzsch:2007} to model flows in DNA microarrays. The orbits of this map, for a parameter value $r>0$ and randomized initial point $(x_0,y_0) \in [0,1] \times [0,1]$, can be calculated by the following rule:
\begin{equation} \label{eqn:linked}
\begin{cases} 
      x_{n+1} = x_n + ry_n(1-y_n) & \text{ mod } 1 \\
     y_{n+1} = y_n + rx_{n+1}(1 - x_{n+1}) & \text{ mod } 1
   \end{cases}
\end{equation}

We generated ten iterations for each $r$ value using different initial points $(x_0,y_0)$. 
\begin{figure}
  \centering
\minipage{0.32\textwidth}
    \includegraphics[width=\linewidth]{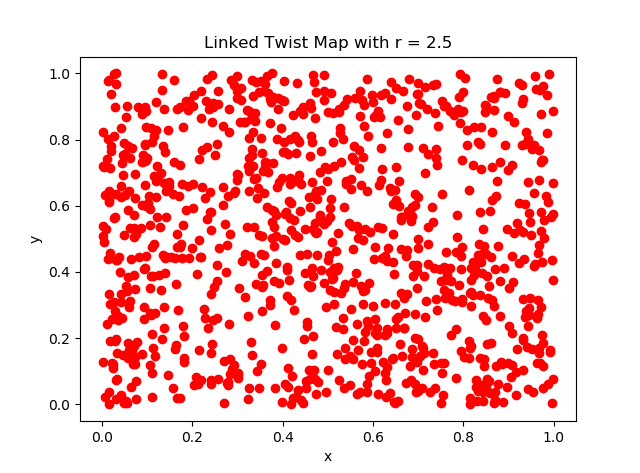}

\endminipage\hfill
\minipage{0.32\textwidth}
	 \includegraphics[width=\linewidth]{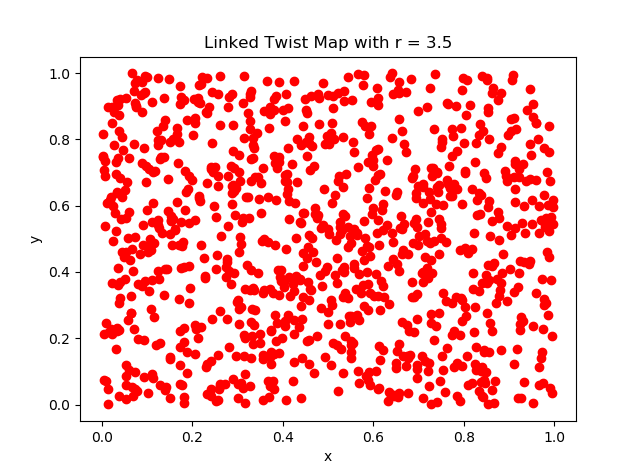}

\endminipage\hfill
\minipage{0.32\textwidth}
	 \includegraphics[width=\linewidth]{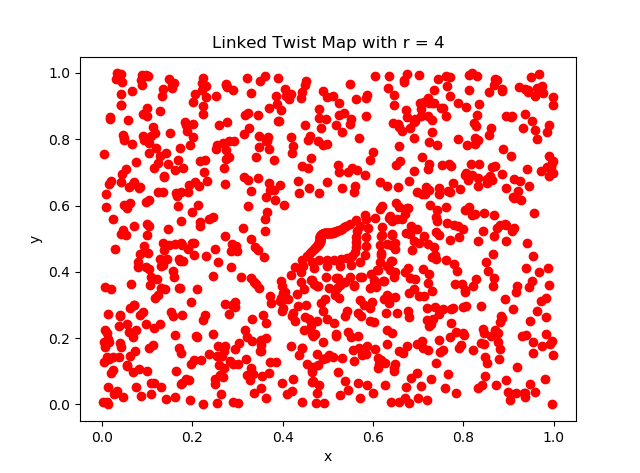}

\endminipage\hfill
\minipage{0.32\textwidth}
	 \includegraphics[width=\linewidth]{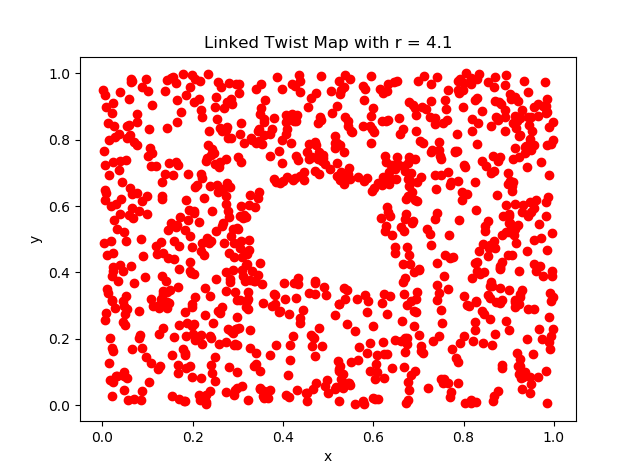}

\endminipage
\minipage{0.32\textwidth}
	 \includegraphics[width=\linewidth]{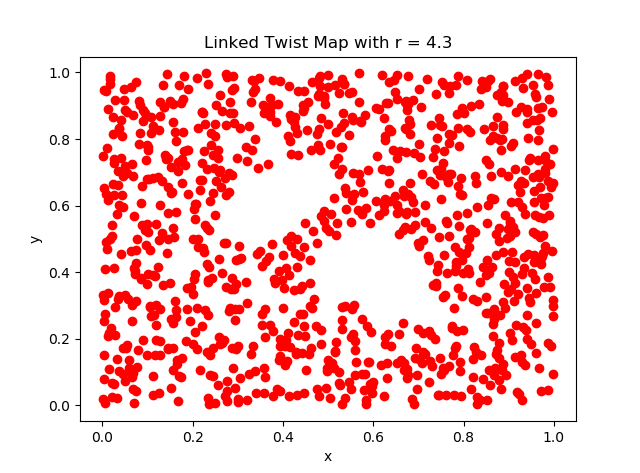}

\endminipage\hfill

\caption{Sample linked-twist map point clouds for parameter values $r = 2.5, 3.5, 4, 4.1,$ and $4.3$. Each point cloud contains 1000 data points. }\label{fig:linkedclouds}
\end{figure}
Figure \ref{fig:linkedclouds} contains sample point clouds for parameter values $r = 2.5, 3.5, 4, 4.1,$ and $4.3$ with a random initial point.
For each point cloud and its corresponding 1-dimensional PD we computed the persistence vector, the unstable interconnectivity vector and the stable interconnectivity vector. A value of $\delta = 0.5$ was used in the computation of the stable interconnectivity vector. 
\begin{figure}[hbt!]
  \centering
\minipage{0.6\textwidth}
    \includegraphics[width=\linewidth]{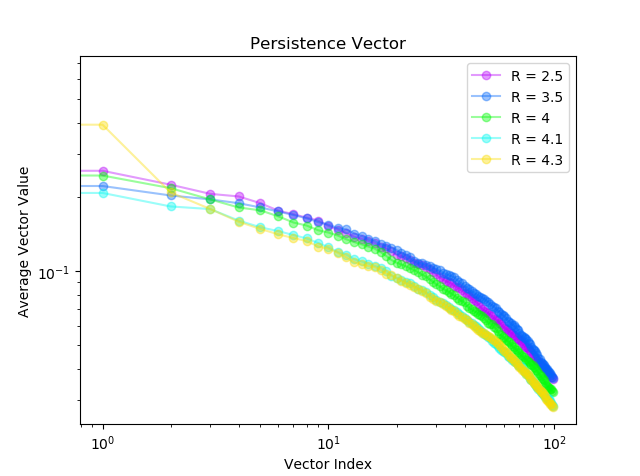}

\endminipage\hfill
\minipage{0.6\textwidth}
	 \includegraphics[width=\linewidth]{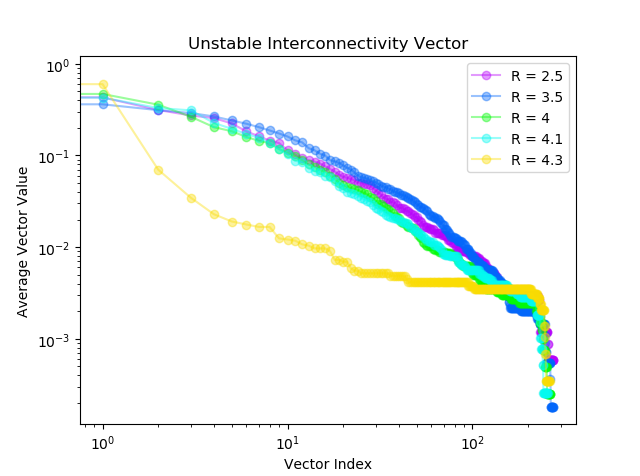}

\endminipage\hfill
\minipage{0.6\textwidth}
	 \includegraphics[width=\linewidth]{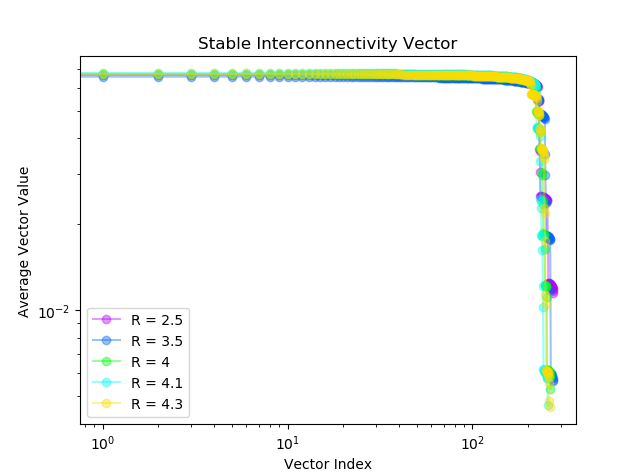}

\endminipage\hfill

\caption{Average vector values for each parameter value $r$ versus vector index in log-log scale beginning with the persistence vector, followed by the unstable interconnectivity vector, ending with the stable interconnectivity vector.}\label{fig:linkedvectors}
\end{figure}

\begin{figure}[hbt!]
  \centering
\minipage{\textwidth}
    \includegraphics[width=0.499\linewidth]{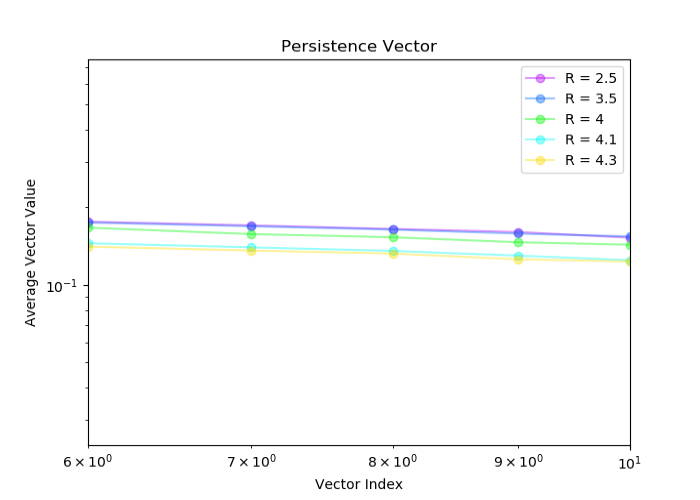}
\includegraphics[width=0.499\linewidth]{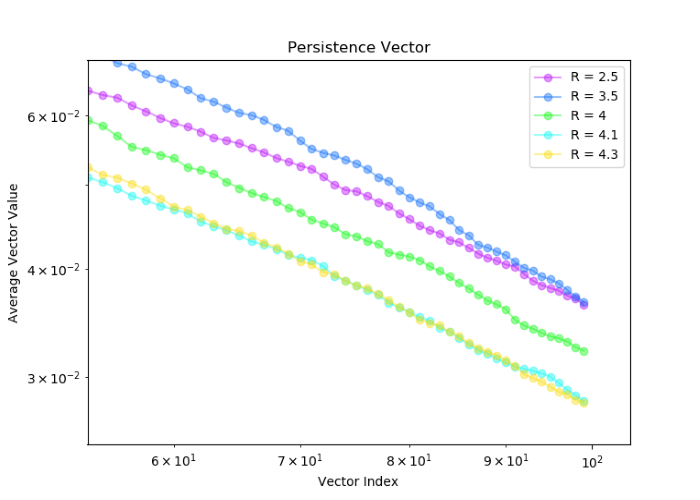}

\endminipage\hfill
\minipage{\textwidth}
	 \includegraphics[width=0.499\linewidth]{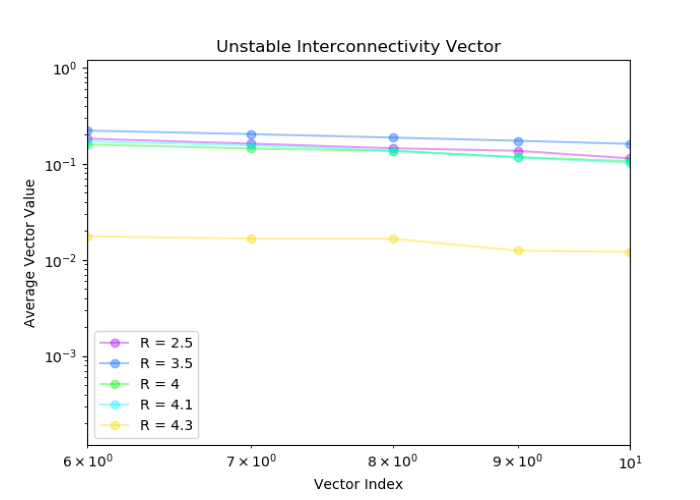}
	\includegraphics[width=0.499\linewidth]{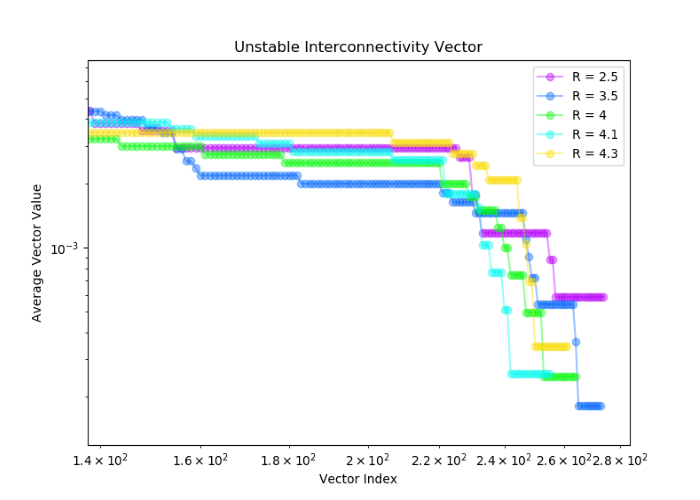}

\endminipage\hfill
\minipage{\textwidth}
	 \includegraphics[width=0.499\linewidth]{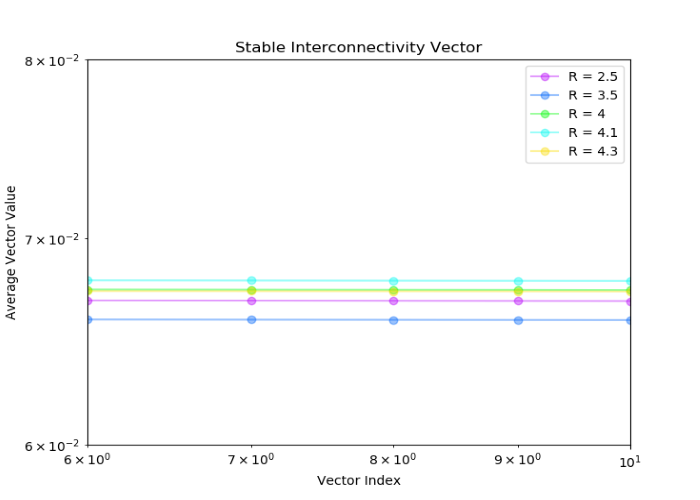}
	 \includegraphics[width=0.499\linewidth]{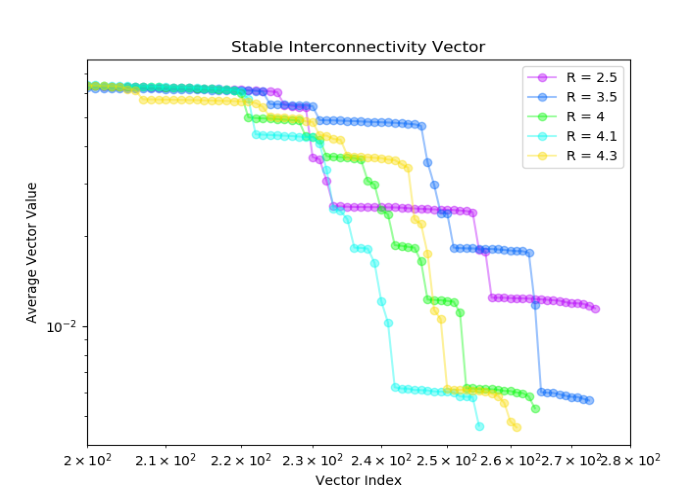}

\endminipage\hfill

\caption{Zoomed low (left) and high (right) index average vector values versus vector index in log-log scale for each parameter value $r$ beginning with the persistence vector, followed by the unstable interconnectivity vector, ending with the stable interconnectivity vector.}\label{fig:linkedvectorslowandhigh}
\end{figure}

In Figure \ref{fig:linkedvectors} there are plots of the average vector value versus the vector index for each of the persistence, unstable and stable interconnectivity vectors, from top to bottom. The most striking observation in the first two graphs in this figure is the distinction between the vector values for the Linked-Twist map with parameter $r = 4.3$ (in yellow) particularly in the low indices (towards the left of each plot) with the (unstable) interconnectivity vector. In particular, the unstable interconnectivity vector plot shows a dramatic difference between the map with parameter $r = 4.3$ and the other maps with different parameters. If we compare this distinction with the underlying point clouds back in Figure \ref{fig:linkedclouds}, we notice that the map with parameter $r = 4.3$ has two distinct holes in the center of the point cloud.  
The distinctive power seems to increase when the unstable interconnectivity vector is stabilized into the stable interconnectivity vector (the bottom figure). Figure \ref{fig:linkedvectorslowandhigh} shows the zoomed plots of those three vectors. The left column shows the average vector values versus low vector index and the right column the same versus high vector index.  In the zoomed plot, it is clear that the stable interconnectivity vector shows much more clearly the distinctions among the 5 different cases of $r$ values, particularly with the lower vector indices. 

It is interesting to observe the decay patterns of the average vector values when the vector index is high. The vector values of high vector index are related to the small scale structure of the given point cloud as the high index entries are the smallest in the vector. The high vector index entries correspond to persistence points with small persistence located close to the diagonal of the PD. In this particular example, the average vector values of high vector indices, both the unstable and stable interconnectivity vectors show sharp cascading decay patterns while the persistence vector does not.

\subsection{Handwriting Example}
For this example we considered 26 hand drawn lower case Roman alphabet characters {\it a} through {\it z}. Our data consists of three sets of 26 characters written by two different authors using a pressure sensitive stylus. Each character was turned into a 100 by 100 pixel greyscale image and the color intensity was measured in each pixel. The intensity ranges from 0 (black) to 255 (white). 
\begin{figure}[hbt!]
  \centering
    \includegraphics[width=0.365\textwidth]{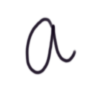}
    \includegraphics[width=0.49\textwidth]{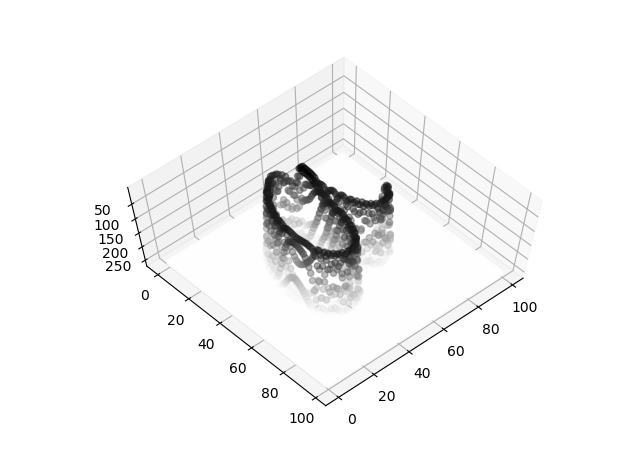}
\caption{Hand drawn letter a (left) sized to 100x100 pixels and its point cloud generated by measuring the intensity of each pixel (right).}\label{fig:lettersa}
\end{figure}
In Figure \ref{fig:lettersa} we have provided a sample letter and its corresponding intensity point cloud. For each point cloud (3 for each character) we calculated the persistence vector, unstable and stable interconnectivity vectors for the 1-dimensional PH.
 Each character's corresponding unstable interconnectivity, stable connectivity, and persistence vectors were then averaged across the three sets of characters. 
\begin{figure}[hbt!]
  \centering
\minipage{0.6\textwidth}
    \includegraphics[width=\linewidth]{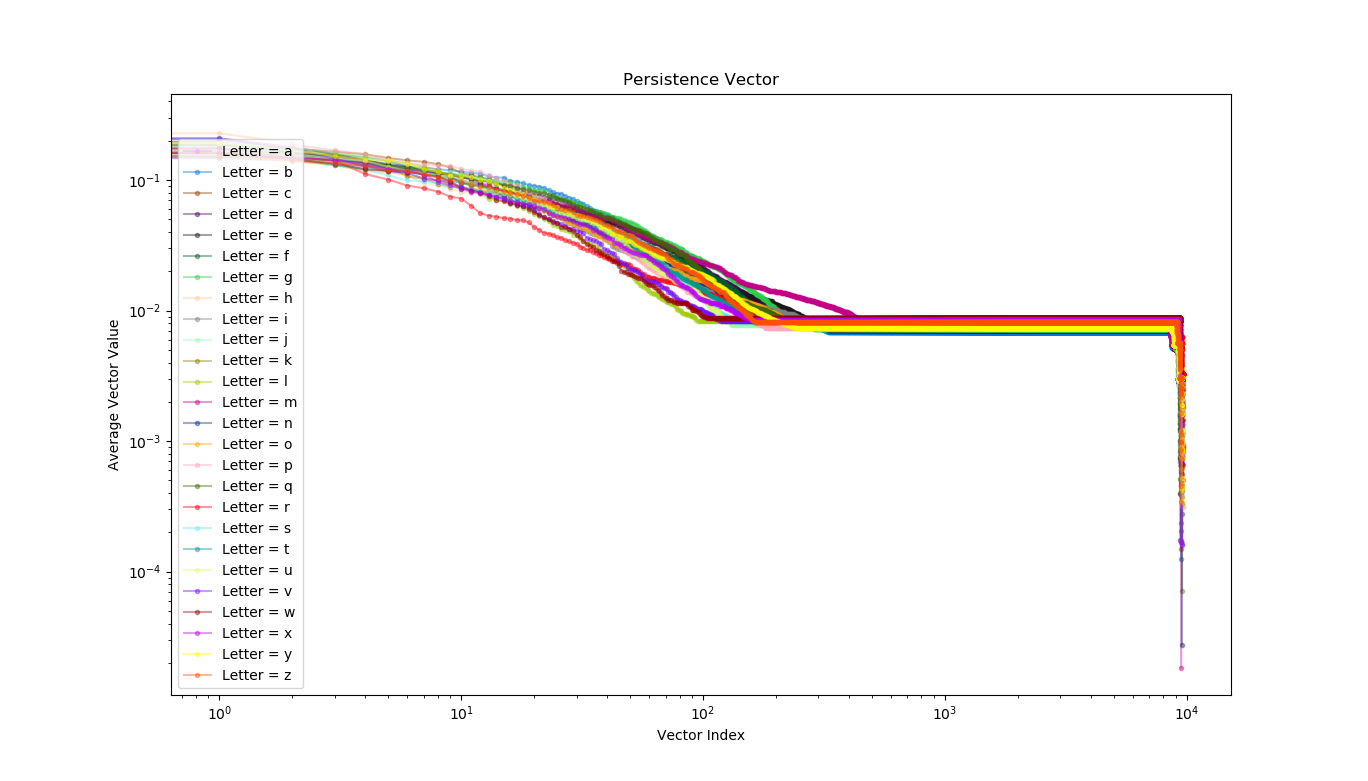}

\endminipage\hfill
\minipage{0.6\textwidth}
	 \includegraphics[width=\linewidth]{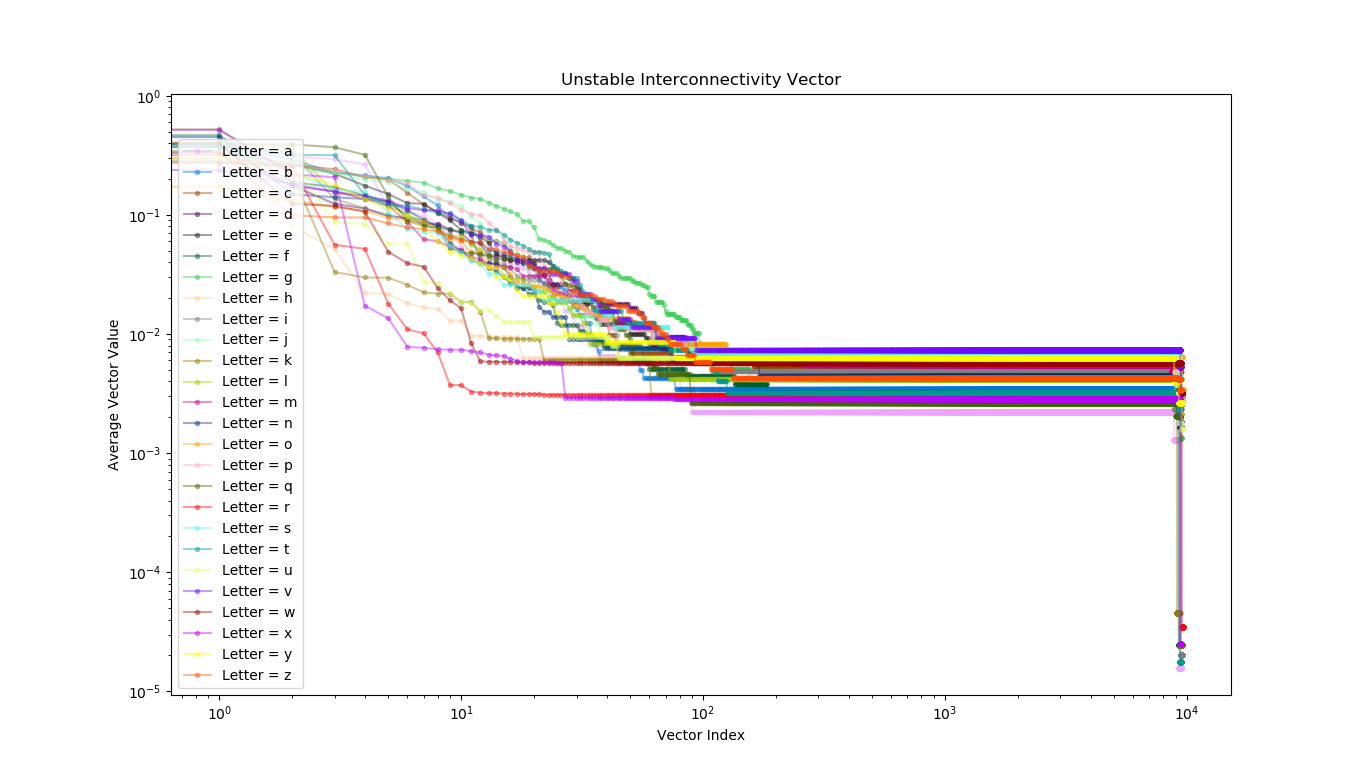}

\endminipage\hfill
\minipage{0.6\textwidth}
	 \includegraphics[width=\linewidth]{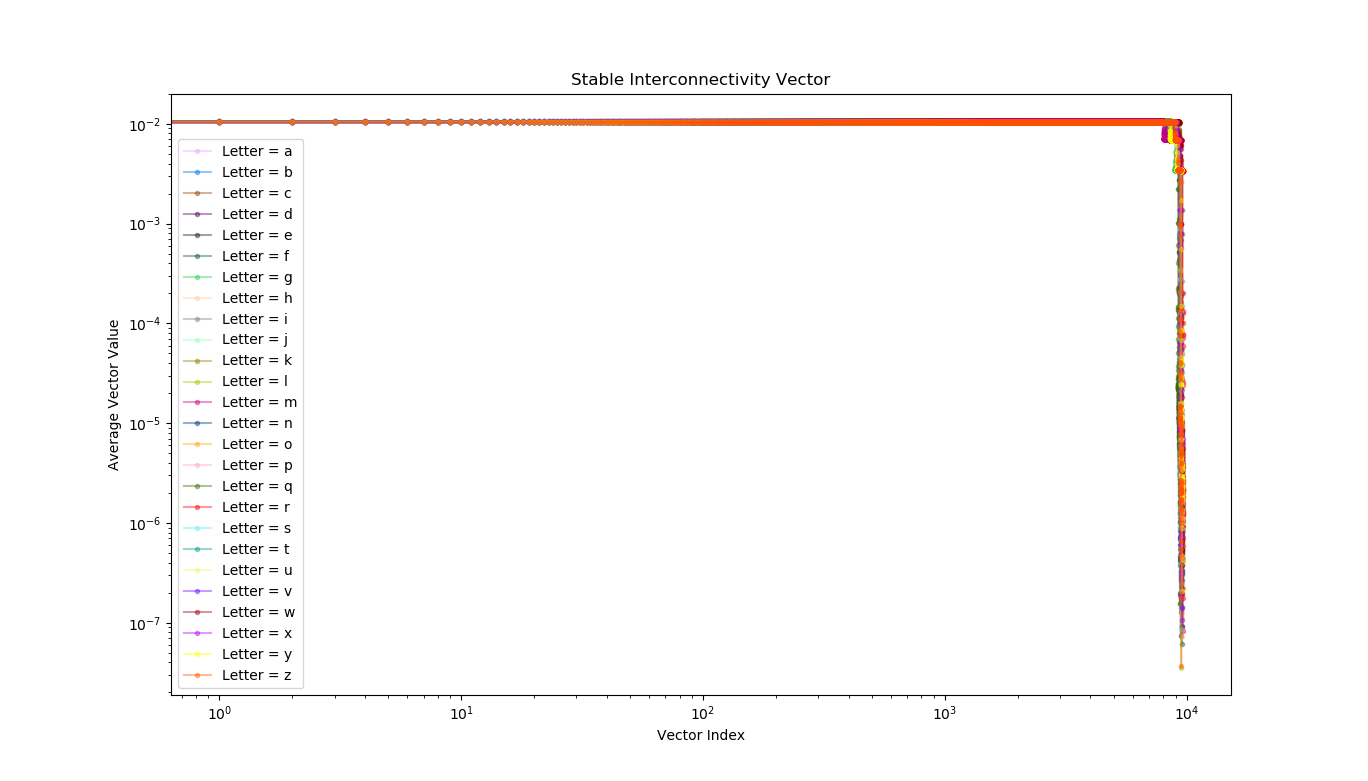}

\endminipage\hfill

\caption{Log-log plots of the average persistence vector, the average unstable interconnectivity vector, and the average stable interconnectivity vector for the 26 Roman characters.}\label{fig:lettersaverage}
\end{figure}

\begin{figure}[hbt!]
  \centering
\minipage{\textwidth}
    \includegraphics[width=0.499\linewidth]{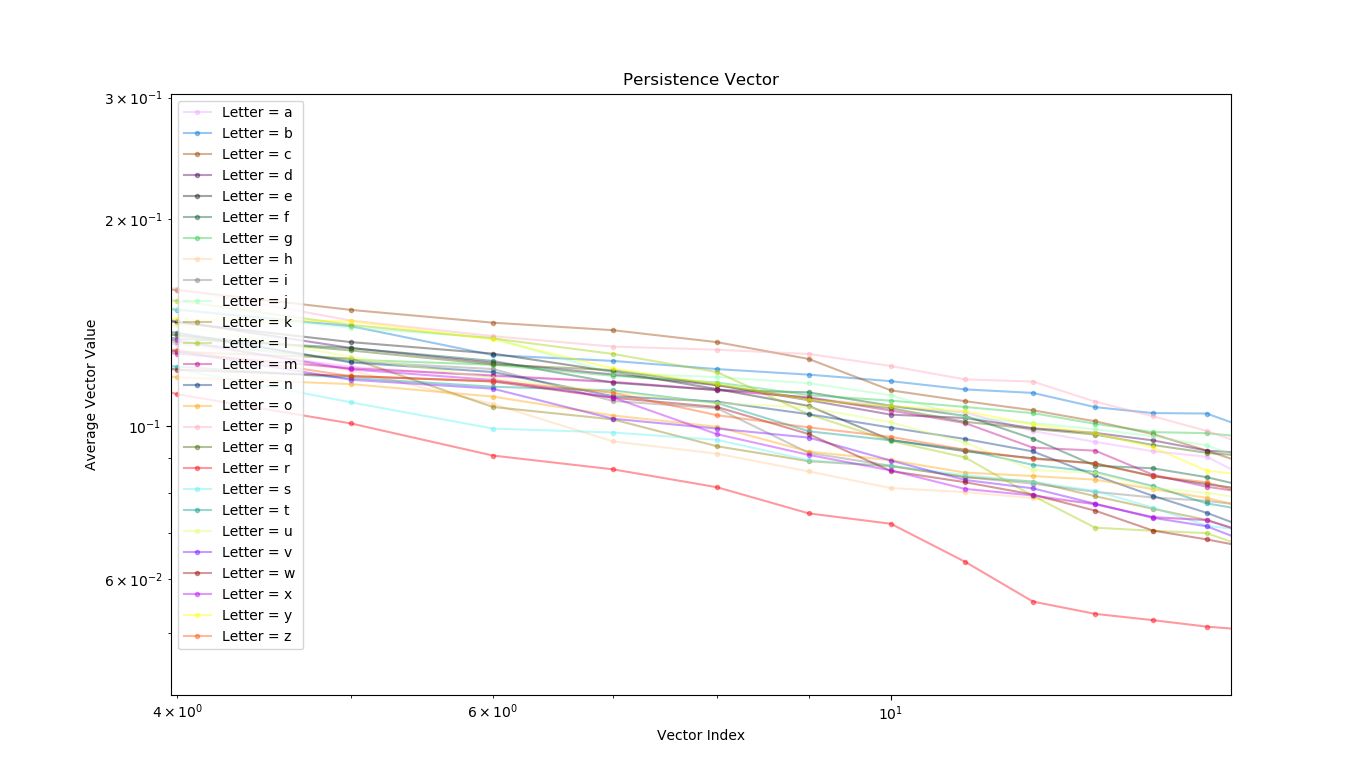}
\includegraphics[width=0.499\linewidth]{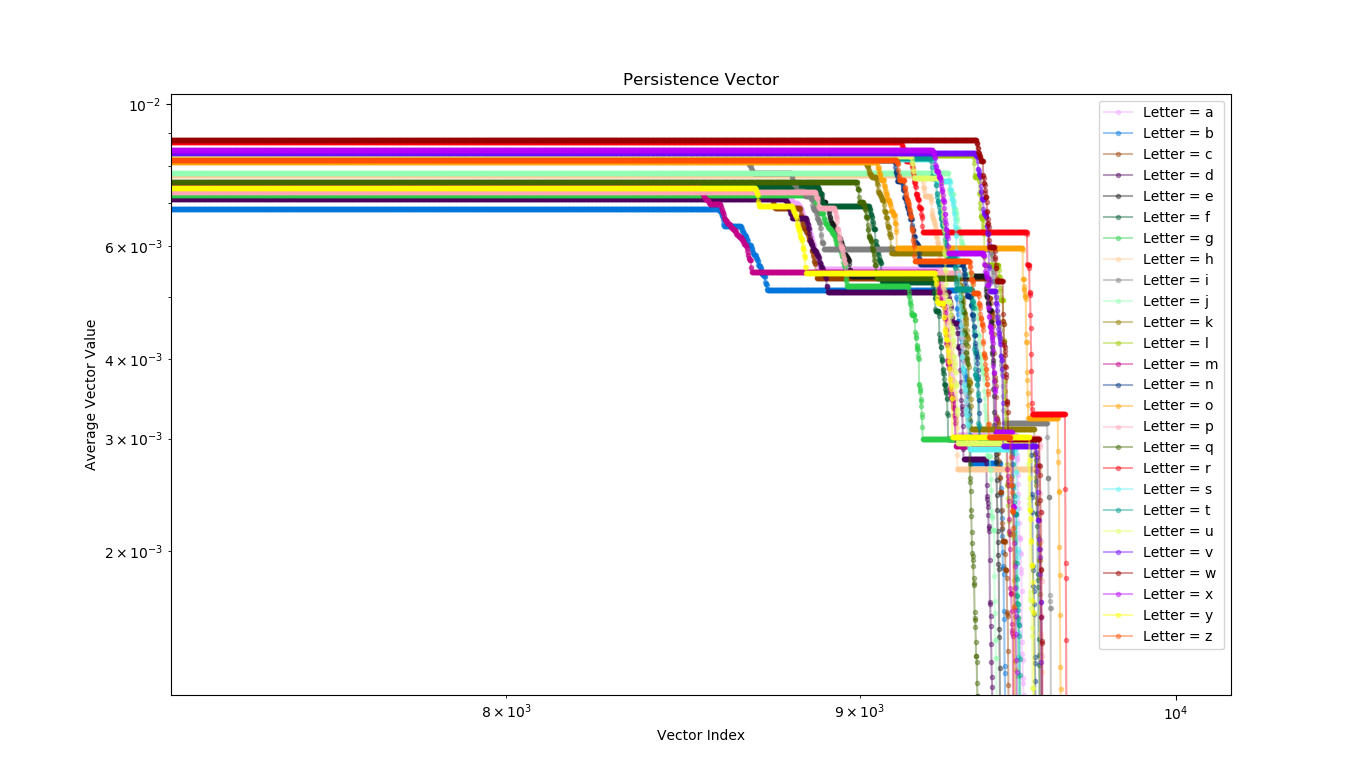}

\endminipage\hfill
\minipage{\textwidth}
	 \includegraphics[width=0.499\linewidth]{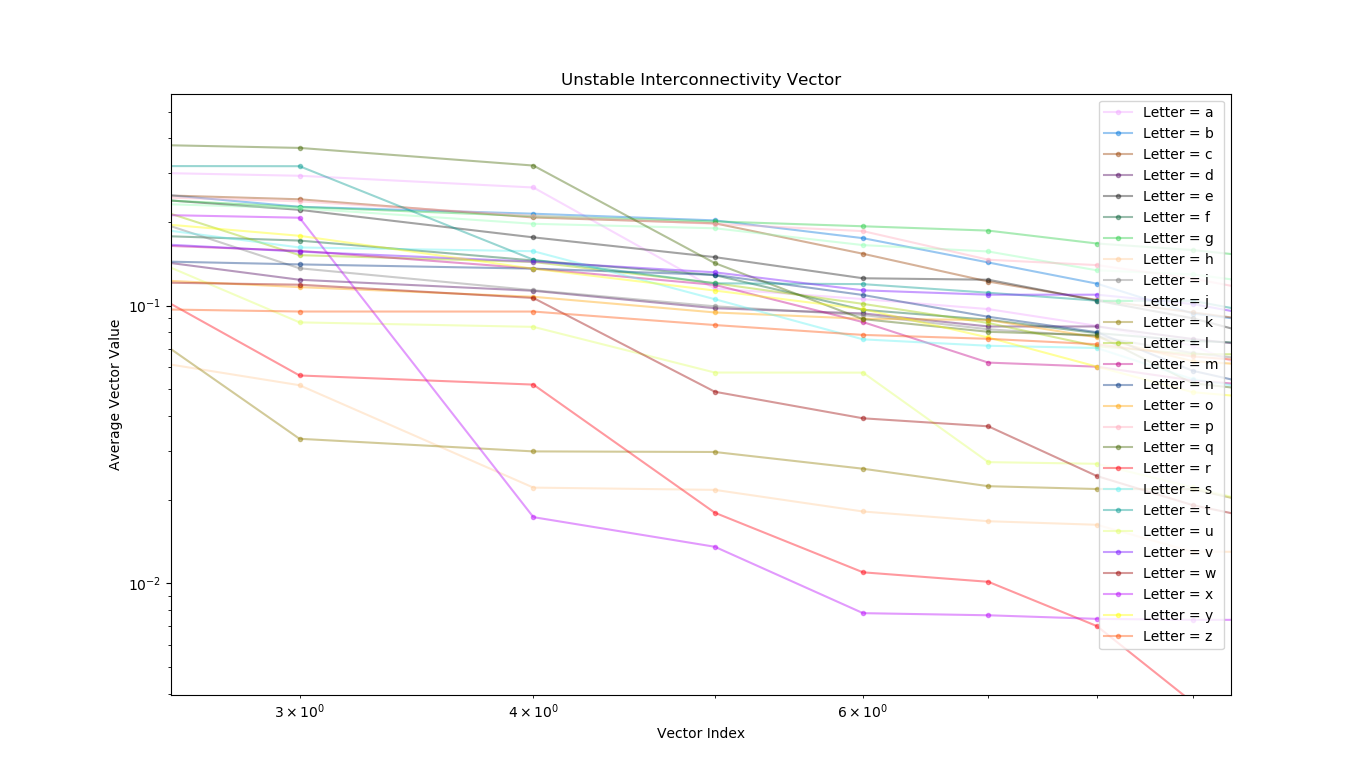}
	\includegraphics[width=0.499\linewidth]{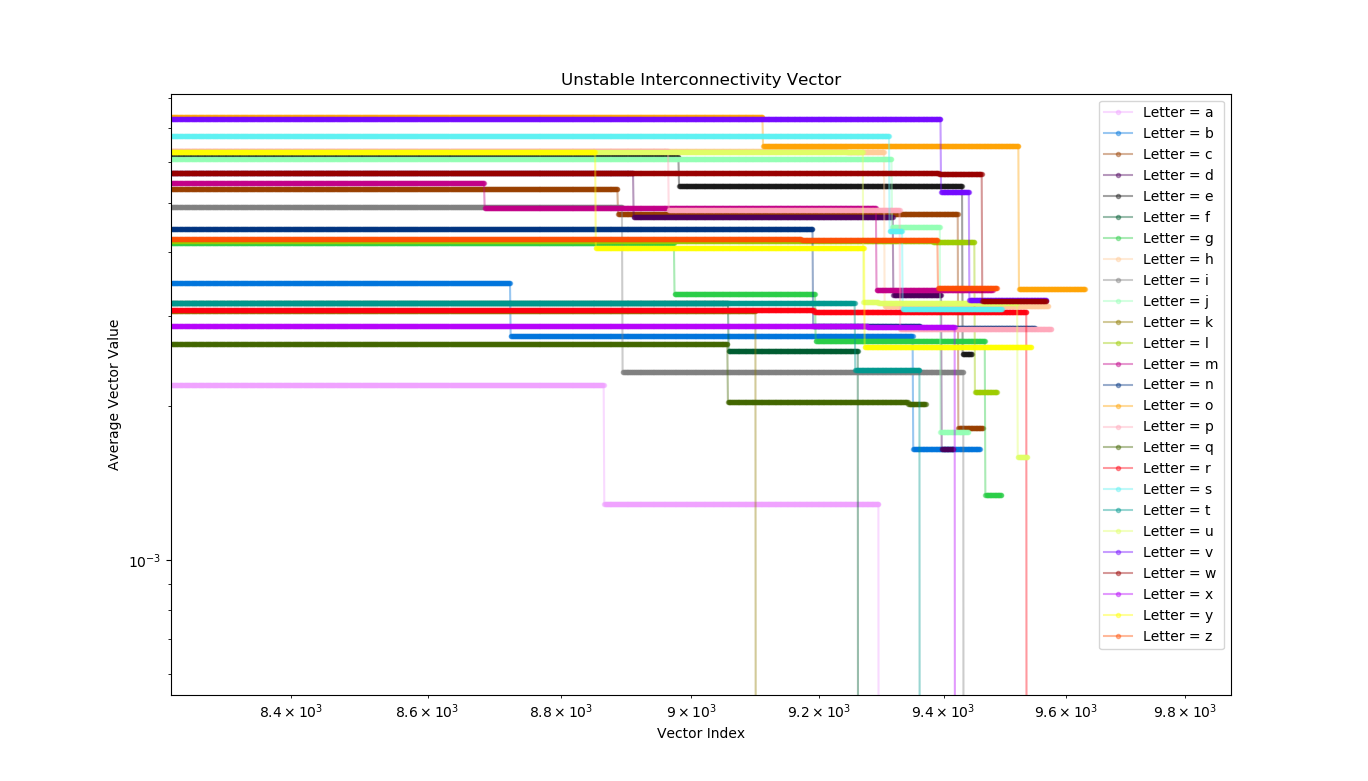}

\endminipage\hfill
\minipage{\textwidth}
	 \includegraphics[width=0.499\linewidth]{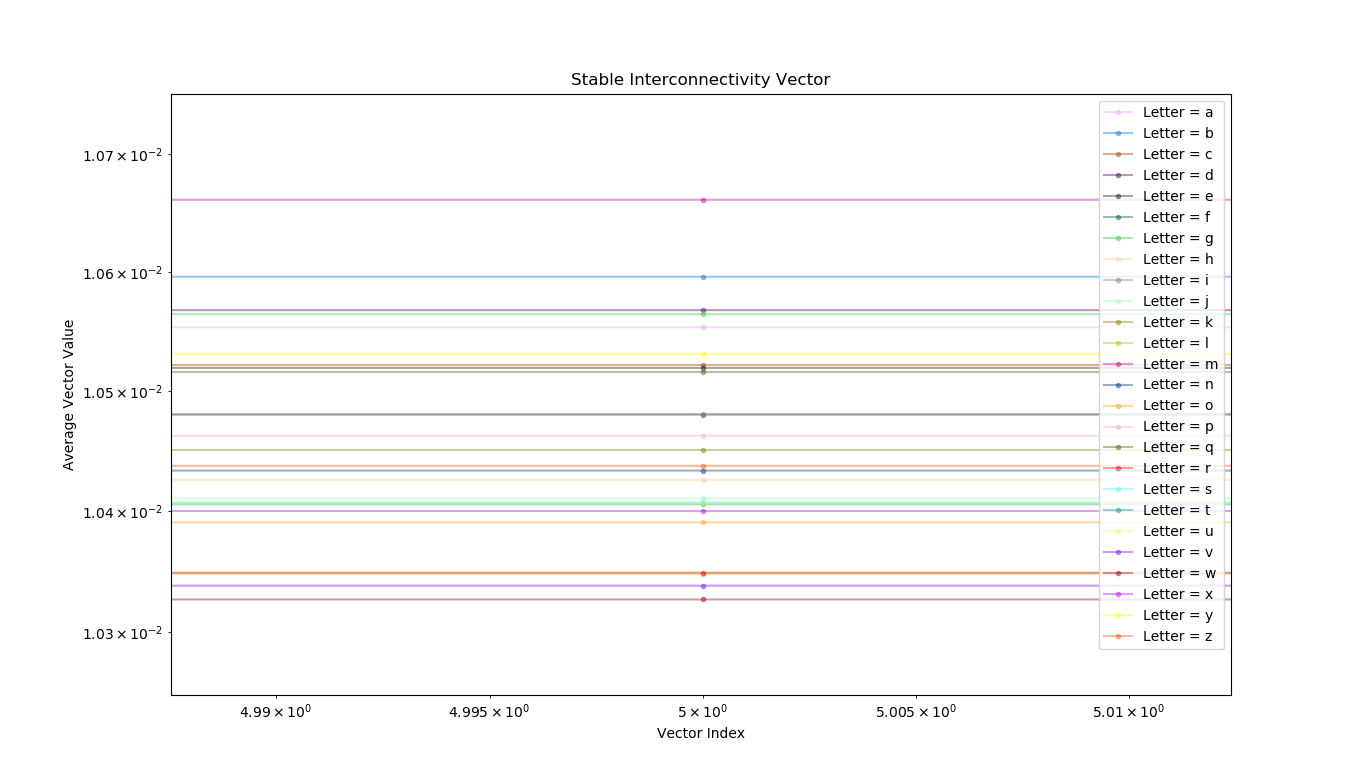}
	 \includegraphics[width=0.499\linewidth]{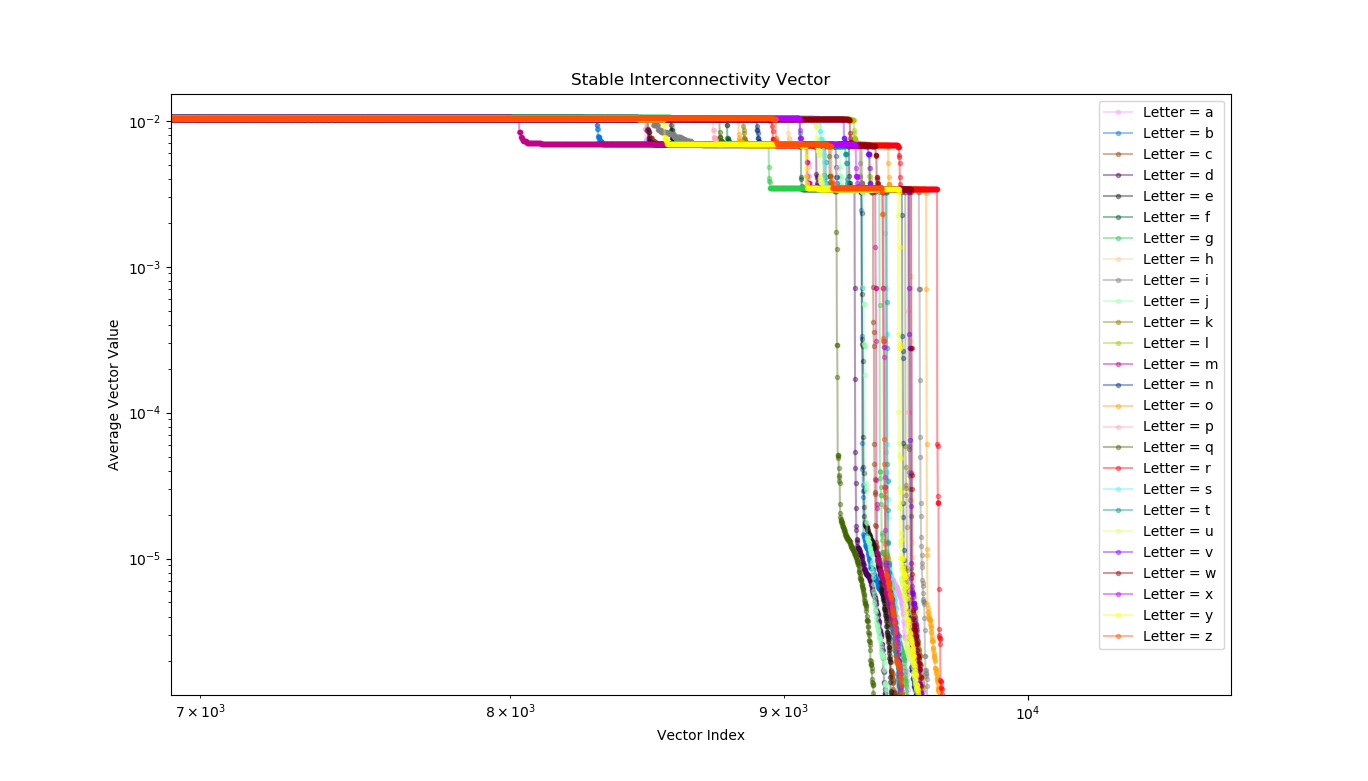}

\endminipage\hfill
\caption{Zoomed low (left) and high (right) index average vector values versus vector index in log-log scale for each letter beginning with the persistence vector, followed by the unstable interconnectivity vector, ending with the stable interconnectivity vector for the 26 Roman characters.}\label{fig:lettervectorslowandhigh}
\end{figure}

Looking at the unstable interconnectivity vector values plotted against the entry index we see in Figure \ref{fig:lettersaverage} that there is no way to naturally separate the characters. For the low indices (left on the graph) there are multiple crossings however for the mid-range and large indices (middle and right on the graph) there are relatively few crossings.

Figure \ref{fig:lettervectorslowandhigh} shows the zoomed average vector values versus vector index in log-log scale for the persistence, unstable and stable interconnectivity vectors from top to bottom. The left column shows the average vector values for lower vector indices while the right column shows the average vector values for higher vector indices. As the figure clearly shows the stable interconnectivity vector shows the most distinctions among characters with lower vector indices. We observe multiple crossings in the figure for the persistence and unstable interconnectivity vectors but no crossings for the stable interconnectivity vector, indicating the potential power of the stable interconnectivity vector in applications, particularly for classification problems.


\section{Concluding Remarks}
Topological data analysis, based on persistent homology, has recently gained much attention in the scientific community and has proved to be highly useful in various applications. The outcome of persistent homology is realized as so-called the persistence barcode and persistence diagram when applied to applications. Their direct raw forms, however, are not suitable to incorporate into a typical machine learning workflow and proper transformation methods are necessary including the kernel and vectorization methods. The kernel methods are known to be computationally expensive. In this paper, we introduced a new vectorization method, called the interconnectivity vector method. The proposed interconnectivity vector method is different from the traditional persistence vector in that it considers not only the persistence but also the interconnectivity between persistence points. By adding this interconnectivity feature to the vectorization, the proposed method yields more distinctive ability to classify topologically different data sets. In this paper, we showed that the proposed interconnectivity vector is unstable and provided a stabilizing method to form the stable interconnectivity vector. Numerical examples showed that the interconnectivity vector performs better than the persistence vector and the stable interconnectivity vector performs the best of the three vectors examined. While conducting the numerical examples, we observed interesting decay patterns of the interconnectivity vector depending on the scale  considered. In our future research, we will further investigate those patterns and study the implications in real applications. We will also incorporate the proposed interconnectivity vector into the machine learning workflow in real applications. 
\clearpage

\nocite{*}

\bibliography{Interconnectivity_Vector_Paper_AMS_format}{}
\bibliographystyle{plain}


\end{document}